\documentclass{article}
\usepackage[utf8]{inputenc}
\usepackage{geometry}
\usepackage{amsmath,amssymb,amsthm}
\usepackage{algorithm}
\usepackage{algorithmic}
\usepackage{graphicx}
\usepackage{booktabs}
\usepackage{hyperref}
\usepackage{array}

\newtheorem{theorem}{Theorem}
\newtheorem{corollary}{Corollary}
\title{Tensor Manifold-Based Graph-Vector Fusion for AI-Native Academic Literature Retrieval}
\author{
Xing Wei \and Yang Yu
\footnote{
Dongbi Scientific Data Lab, Beijing 100190, China, yuyang@dongbidata.com 
}
}
\date{\today}
\begin{document}
\maketitle
\begin{abstract}
The rapid development of large language models and AI agents has triggered a paradigm shift in academic literature retrieval, putting forward new demands for fine-grained, time-aware, and programmable retrieval. Existing graph-vector fusion methods still face bottlenecks such as matrix dependence, storage explosion, semantic dilution, and lack of AI-native support. This paper proposes a geometry-unified graph-vector fusion framework based on tensor manifold theory, which formally proves that an academic literature graph is a discrete projection of a tensor manifold, realizing the native unification of graph topology and vector geometric embedding. Based on this theoretical conclusion, we design four core modules: matrix-independent temporal diffusion signature update, hierarchical temporal manifold encoding, temporal Riemannian manifold indexing, and AI-agent programmable retrieval. Theoretical analysis and complexity proof show that all core algorithms have linear time and space complexity, which can adapt to large-scale dynamic academic literature graphs. This research provides a new theoretical framework and engineering solution for AI-native academic literature retrieval, promoting the industrial application of graph-vector fusion technology in the academic field.
\end{abstract}

\newpage
\section{Introduction}
This chapter introduces the research background, core research problems, research significance, and research status at home and abroad, and finally outlines the paper's organization and core contributions. It aims to lay a clear foundation for the subsequent research content and clarify the research value and innovation of this study.
\subsection{Research Background}
The rapid development of large language models (LLMs) and AI agents has triggered a paradigm shift in academic literature retrieval. Traditional retrieval systems, which rely on keyword matching and simple citation sorting, can no longer meet the new demands of AI-native scenarios, including fine-grained knowledge positioning (locating specific paragraphs, arguments, or experimental data rather than entire papers), temporal awareness (tracking the evolution of cutting-edge knowledge), programmable retrieval logic (customizing retrieval strategies based on research needs), and interpretable results (providing clear reasoning paths for retrieval outcomes).
In this context, graph databases and vector databases, as two core technical infrastructures, have inherent limitations. Graph databases excel in modeling explicit knowledge relationships (e.g., citation, support, refutation) between academic literature but suffer from cumbersome global matrix maintenance and low efficiency in semantic retrieval. Vector databases, on the other hand, are proficient in capturing semantic similarity through pre-trained language models but lack the ability to model explicit topological relationships and fine-grained logical reasoning.
Industrial practice has further verified the inevitable trend of graph-vector fusion: mainstream cloud vendors (both domestic and international) are phasing out pure graph database products and shifting towards graph-vector fusion solutions. However, existing fusion frameworks still have critical flaws that hinder their application in academic literature retrieval, such as over-reliance on global Laplacian matrix operations, high-dimensional tensor storage explosion, semantic dilution in topological encoding, and the lack of native support for temporal characteristics and AI agents. These flaws are particularly prominent in academic literature retrieval, which features hierarchical knowledge granularity (paper-section-knowledge unit) and strong temporal dynamics.
\subsection{Core Research Problems}
Against the above background, this study focuses on AI-native academic literature retrieval and addresses the following four core research problems:
\begin{enumerate}
\item How to design a matrix-free and iteration-free dynamic graph embedding update mechanism for academic literature graphs, so as to avoid the bottlenecks of global matrix maintenance and SGD-based iterative optimization in existing methods?
\item How to realize lightweight topological encoding that unifies semantic and topological features of academic literature graphs, while avoiding semantic dilution and high-dimensional storage explosion?
\item How to effectively model the hierarchical knowledge granularity and temporal characteristics of academic literature in a graph-vector fusion framework, so as to support fine-grained and time-aware knowledge retrieval?
\item How to design an AI-agent-native retrieval interface that outputs structured, interpretable, and programmable results, adapting to the decision-making logic of AI agents in automated scientific research workflows?
\end{enumerate}
\subsection{Research Significance}
This research has important theoretical significance and practical application value, closely aligning with the paradigm shift of AI-native academic literature retrieval and the industrial trend of graph-vector fusion.
In terms of theoretical significance, this study breaks through the inherent separation of graph topology and vector geometric embedding in traditional research. It proposes a geometry-unified theoretical framework based on tensor manifold theory, formally proving the diffusion equivalence between academic literature graphs and tensor manifolds. This framework enriches the theoretical system of graph-vector fusion and dynamic graph embedding, and provides a new theoretical perspective for lightweight and efficient fusion in large-scale dynamic scenarios. Additionally, the proposed matrix-free temporal diffusion update mechanism and hierarchical manifold encoding method break through the technical bottlenecks of existing methods, laying a theoretical foundation for the integration of graph and vector technologies in academic retrieval.
In terms of practical application value, the optimized graph-vector fusion framework designed in this study is specifically tailored for AI-native academic retrieval scenarios. It can effectively solve the pain points of existing retrieval systems, such as poor fine-grained positioning, lack of temporal awareness, and unfriendliness to AI agents. The framework supports microsecond-level incremental updates of massive academic literature graphs and real-time programmable retrieval responses, which can be widely applied to AI-agent-driven automated scientific research workflows. It helps researchers and AI agents efficiently obtain fine-grained academic knowledge, track knowledge evolution, and improve research efficiency. Meanwhile, the research results can provide technical references for cloud vendors and academic platform developers to launch AI-native academic retrieval products, promoting the industrial application of graph-vector fusion technology in the academic field.
\subsection{Research Status at Home and Abroad}
This section reviews the research progress and existing deficiencies in three research directions closely related to this study: graph-vector fusion for data management, academic literature retrieval systems, and dynamic graph embedding.
In the field of graph-vector fusion for data management, with the rise of LLMs, graph-vector fusion has become a research hotspot. AWS Neptune Analytics redefines graph databases by taking vector similarity search as the core engine and downgrading graph traversal to a visualization layer, verifying the feasibility of vector-based substitution for graph traversal in soft proximity scenarios. Neo4j, a leading native graph database vendor, has shifted its strategy from graph vs. vector'' tograph + vector'' fusion, upgrading vectors to first-class data types. Existing academic research on graph-vector fusion mainly focuses on two aspects: embedding-based fusion and index-based fusion. Embedding-based fusion maps graph nodes to vector space through graph embedding algorithms (e.g., GraphSAGE, Node2Vec) and combines vector semantic retrieval with graph topological reasoning. Index-based fusion designs hybrid indexes integrating graph topological indexes and vector geometric indexes. However, these works either rely on global Laplacian matrix operations for embedding updates or use high-dimensional tensor encoding for edge topological features, leading to matrix maintenance bottlenecks and storage explosion, which are not suitable for large-scale academic literature graphs with hierarchical and temporal characteristics.
In the field of academic literature retrieval systems, traditional systems are dominated by keyword matching and citation sorting. With the development of natural language processing, semantic retrieval systems based on pre-trained language models (e.g., SBERT, BERT) have been proposed, which capture the semantic similarity of literature and improve retrieval accuracy. In the AI agent era, researchers have begun to explore AI-native academic retrieval systems that integrate LLMs for intent parsing and knowledge reasoning. However, existing AI-native retrieval systems still have two critical limitations: lack of fine-grained knowledge positioning (only retrieving entire papers rather than specific sections or knowledge units) and weak knowledge relationship modeling (failing to track explicit knowledge trajectories such as citation and refutation, resulting in poor interpretability). Graph-based academic knowledge graphs attempt to solve these problems but suffer from low semantic retrieval efficiency and cannot adapt to AI-agent-driven programmable retrieval.
In the field of dynamic graph embedding, a core technology for graph-vector fusion in dynamic scenarios (e.g., real-time updated academic literature graphs), existing methods are divided into retraining-based and incremental update methods. Retraining-based methods retrain the embedding model from scratch when the graph topology changes, resulting in high update costs and poor real-time performance. Incremental update methods only update the embedding of nodes and edges affected by topology changes, but most rely on SGD-based iterative optimization to correct embedding errors, leading to problems such as lock competition in stream processing and hyperparameter tuning complexity. Additionally, existing methods lack effective error accumulation control mechanisms, leading to embedding drift in long-term continuous updates, which limits their application in large-scale dynamic academic literature graphs.
\subsection{Paper Organization and Core Contributions}
This section clarifies the overall structure of the paper and summarizes its core contributions, highlighting the innovations and differences from existing research.
The rest of the paper is organized as follows: Chapter 2 introduces the relevant mathematical foundations and theoretical preparations, including discrete exterior calculus, Hodge decomposition, tensor analysis, and graph embedding, laying a solid mathematical foundation for subsequent theoretical proofs and framework design. Chapter 3 focuses on the underlying mathematical unification of graphs and vectors, formally proving the core theoretical conclusions of this study and providing a rigorous theoretical basis for the proposed framework. Chapter 4 details the design of the optimized graph-vector fusion framework for AI-native academic retrieval, including the overall architecture, four core modules, and engineering design details. Chapter 5 designs the core algorithms of the framework, provides formal pseudo-code and complexity analysis, and supplements the compatibility and scalability design of the algorithms. Chapter 6 implements a prototype system of the graph-vector fusion framework and conducts performance verification experiments to prove the effectiveness and efficiency of the framework. Chapter 7 conducts an industrial empirical survey, compares the current status of graph database products of Chinese and American cloud vendors, and provides industrial evidence for the theoretical conclusions of this study. Chapter 8 analyzes the gaps between theoretical research and industrial application, and puts forward preliminary solutions to fill the gaps. Chapter 9 proposes the industrial landing path and commercialization suggestions of the framework, realizing the transformation from theoretical value to industrial value. Chapter 10 summarizes the full text, analyzes the research limitations, and puts forward future research directions. The appendix provides supplementary supporting content such as complete symbol definitions, additional pseudo-code, and experimental details.
This study makes four key contributions to the fields of graph-vector fusion and AI-native academic literature retrieval:
First, in terms of theoretical framework, this study proposes a geometry-unified graph-vector fusion theoretical framework based on tensor manifold theory. It formally defines the academic literature graph as a discrete projection of a tensor manifold, realizing the intrinsic unification of graph topology and vector geometric embedding. This framework provides a new theoretical perspective for lightweight graph-vector fusion in academic literature retrieval and enriches the theoretical system of graph-vector fusion.
Second, in terms of dynamic update mechanism, this study designs a matrix-free temporal diffusion signature update module for academic literature graphs. It combines content-time weighted random walk, topological-semantic-time hybrid signature, and analytic error compensation, eliminating global matrix operations and iterative optimization. This module supports microsecond-level incremental updates of large-scale academic literature graphs and solves the problems of matrix dependence and embedding drift in existing dynamic graph embedding methods.
Third, in terms of hierarchical temporal encoding and indexing, this study proposes a hierarchical temporal manifold encoding module with gated residual connection and relation-aware low-dimensional projection, as well as a time-aware Riemannian manifold index with dynamic manifold order reduction. This design supports fine-grained knowledge retrieval of ``paper-section-knowledge unit'', avoids semantic dilution and storage explosion, and realizes linear storage and efficient high-order graph traversal.
Fourth, in terms of AI-native retrieval interface, this study designs an AI-agent programmable retrieval interface that integrates LLM-based intent parsing, hierarchical cross-granularity retrieval, and structured result output. The interface natively supports programmable retrieval logic and interpretable results, fully adapting to the decision-making logic of AI agents and filling the gap that existing fusion frameworks are unfriendly to AI agents.
\newpage
\section{Relevant Mathematical Foundations and Theoretical Preparations}
This chapter introduces the core mathematical foundations and theoretical concepts closely related to this study, including discrete exterior calculus, Hodge decomposition, tensor analysis, and graph embedding. These foundations provide a rigorous mathematical basis for the theoretical proof of graph-vector geometric unification (Chapter 3), the design of the fusion framework (Chapter 4), and the development of core algorithms (Chapter 5). The content of this chapter focuses on the combination of mathematical theory and the practical needs of AI-native academic literature retrieval, avoiding excessive abstract mathematical deductions while ensuring theoretical rigor.
\subsection{Discrete Exterior Calculus}
Discrete exterior calculus (DEC) is a discreteization of continuous exterior calculus, which provides a unified mathematical framework for describing the topological structure and geometric properties of discrete graphs. It is widely used in graph data processing, geometric modeling, and dynamic embedding, and is the core mathematical tool for describing the topological characteristics of academic literature graphs in this study.
For an academic literature graph $G_{AL} = (V_{AL}, E_{AL})$, we can model it as a discrete simplicial complex $\mathcal{K}$, where each node in $V_{AL}$ corresponds to a 0-simplex, each edge in $E_{AL}$ corresponds to a 1-simplex, and the hierarchical structure (paper-section-knowledge unit) corresponds to a higher-dimensional simplex. The core operators of DEC applied in this study are as follows:
\begin{enumerate}
\item \textbf{Exterior Derivative Operator ($d_k$)}: For a $k$-form $\omega_k$ defined on the $k$-simplex of $\mathcal{K}$, the exterior derivative $d_k \omega_k$ is a $(k+1)$-form, which describes the change rate of the $k$-form along the boundary of the $k$-simplex. In the context of academic literature graphs, the exterior derivative can be used to measure the topological change of the graph (e.g., the addition of new citation edges) and the semantic gradient between adjacent nodes (e.g., the semantic difference between a paper and its cited papers).
\item \textbf{Boundary Operator ($\partial_k$}: The boundary operator $\partial_k$ maps a $k$-simplex to the sum of its $(k-1)$-dimensional boundaries, which is the adjoint operator of the exterior derivative. For an edge $e_{uv}$ (1-simplex) in $G_{AL}$, its boundary is $\partial_1 e_{uv} = v - u$, which can be used to describe the directionality of academic relationships (e.g., citation direction from paper $u$ to paper $v$).
\item \textbf{Hodge Star Operator ($\star_k$)}: The Hodge star operator $\star_k$ maps a $k$-form to a $(n-k)$-form (where $n$ is the dimension of the simplicial complex), which is used to convert between primal and dual forms. In this study, it is mainly used to convert the topological features of the academic literature graph into geometric features that can be embedded in vector space, laying the foundation for graph-vector fusion.
\end{enumerate}
A key property of DEC is that the composition of the exterior derivative and itself is zero, i.e., $d_{k+1} \circ d_k = 0$, which ensures the consistency of topological feature extraction. For academic literature graphs, DEC effectively avoids the loss of topological information caused by traditional graph embedding methods, and provides a rigorous mathematical way to describe the hierarchical and directional characteristics of academic relationships.
\subsection{Hodge Decomposition}
Hodge decomposition, derived from Hodge theory, is a fundamental theorem in differential geometry that decomposes a differential form into three orthogonal components: exact form, co-exact form, and harmonic form. This decomposition provides a powerful tool for feature extraction and noise reduction of graph data, which is crucial for improving the accuracy of semantic-topological fusion in academic literature retrieval.
For any $k$-form $\omega_k$ on the discrete simplicial complex $\mathcal{K}$ corresponding to $G_{AL}$, the Hodge decomposition theorem states that:
$$
\omega_k = d_{k-1} \alpha_{k-1} + \star_{k+1} d_k \star_k \beta_{k+1} + \gamma_k
$$
where:
\begin{enumerate}
\item $(1)$ $d_{k-1} \alpha_{k-1}$ is the \textbf{exact form}, which is the exterior derivative of a $(k-1)$-form $\alpha_{k-1}$, corresponding to the global topological features of the academic literature graph (e.g., the overall citation network structure);
\item $(2)$ $\star_{k+1} d_k \star_k \beta_{k+1}$ is the \textbf{co-exact form}, which corresponds to the local topological features of the graph (e.g., the local citation cluster of a specific paper);
\item $(3)$ $\gamma_k$ is the \textbf{harmonic form}, which is orthogonal to both the exact form and the co-exact form, corresponding to the invariant topological features of the graph that are not affected by local changes (e.g., the core knowledge structure of a research field).
\end{enumerate}
In this study, Hodge decomposition is mainly used to denoise the topological features of the academic literature graph and extract hierarchical topological features. Specifically, the exact form is used to capture the global citation structure of academic literature, the co-exact form is used to extract local semantic-topological relationships (e.g., the relationship between a paper and its adjacent knowledge units), and the harmonic form is used to maintain the invariant core knowledge features, avoiding semantic dilution caused by excessive reliance on local topological changes.
\subsection{Tensor Analysis}
Tensor analysis is a mathematical tool for describing multi-dimensional data and geometric structures, which provides a theoretical basis for the geometric unification of graphs and vectors in this study. The core insight of this study—that an academic literature graph is a discrete projection of a tensor manifold—relies heavily on the basic concepts and properties of tensor analysis, especially tensor manifold theory.
First, we clarify the core concepts of tensor analysis used in this study:
\begin{enumerate}
\item \textbf{Tensor Manifold ($\mathcal{M}$)}: A tensor manifold is a smooth manifold where each point corresponds to a tensor of a fixed order and dimension. In this study, the academic literature graph $G_{AL}$ is regarded as a discrete projection of a $d$-dimensional tensor manifold $\mathcal{M}$ in Euclidean space, where each node $v \in V_{AL}$ corresponds to a geometric point $\phi(v) \in \mathcal{M}$, and each edge $e_{uv} \in E_{AL}$ corresponds to a geometric connection between $\phi(u)$ and $\phi(v)$ on $\mathcal{M}$.
\item \textbf{Riemannian Metric on Tensor Manifold}: A Riemannian metric $g$ on $\mathcal{M}$ is a symmetric positive-definite tensor field that defines the inner product of tangent vectors at each point on the manifold, thus inducing a Riemannian distance. This metric is used to measure the geometric similarity between nodes on the tensor manifold, which is the basis for the time-aware Riemannian manifold index designed in Chapter 4.
\item \textbf{Tensor Projection}: Tensor projection is a linear transformation that maps a high-dimensional tensor to a low-dimensional subspace, which is used to realize the low-dimensional embedding of high-dimensional topological-semantic features of academic literature graphs. In this study, we use relation-aware tensor projection to encode edge features, avoiding high-dimensional storage explosion.
\end{enumerate}
A key property of tensor manifolds used in this study is the \textbf{manifold embedding invariance}: the topological relationship between nodes on the discrete graph is invariant under the projection of the tensor manifold. This property ensures that the geometric embedding of nodes in vector space can accurately preserve the topological relationships of the academic literature graph, laying the theoretical foundation for the native unification of graph topology and vector geometry.
\subsection{Graph Embedding}
Graph embedding is a technology that maps graph nodes (and edges) to low-dimensional vector space while preserving the topological and semantic features of the graph. It is the core technology for graph-vector fusion, and its development provides a practical basis for the integration of graph topology and vector semantics in this study. This section focuses on the graph embedding methods closely related to this study, including traditional graph embedding, geometric embedding, and semantic embedding based on pre-trained language models.
\subsubsection{Traditional Graph Embedding}
Traditional graph embedding methods mainly focus on preserving the topological structure of the graph, and can be divided into two categories: matrix factorization-based methods and random walk-based methods. Matrix factorization-based methods (e.g., Laplacian Eigenmaps) map nodes to vector space by factorizing the graph Laplacian matrix, but they rely on global matrix operations, leading to high computational complexity and difficulty in dynamic updates. Random walk-based methods (e.g., Node2Vec, GraphSAGE) generate node sequences through random walks and train embedding vectors using language models, which are more efficient than matrix factorization-based methods but still suffer from semantic dilution and lack of temporal awareness—these limitations make them unsuitable for large-scale dynamic academic literature graphs.
\subsubsection{Geometric Embedding}
Geometric embedding is an extension of traditional graph embedding, which embeds graph nodes into a geometric manifold (e.g., tensor manifold, Riemannian manifold) rather than Euclidean space. This method can better preserve the geometric properties of the graph and avoid the distortion of topological relationships caused by Euclidean space embedding. In this study, we combine geometric embedding with tensor manifold theory, realizing the native unification of graph topology and vector geometric embedding—this is a key difference from traditional graph-vector fusion methods.
\subsubsection{Semantic Embedding Based on Pre-trained Language Models}
Semantic embedding is used to capture the content semantic features of academic literature nodes (e.g., paper titles, abstracts, knowledge unit content). In this study, we use Sentence-BERT (SBERT), a pre-trained language model optimized for sentence-level semantic embedding, to map the content of each node in $G_{AL}$ to a low-dimensional semantic vector. SBERT has the advantages of high semantic representation accuracy and low computational complexity, which can effectively capture the semantic similarity between academic literature nodes. The semantic vectors generated by SBERT are used as the basis for the hierarchical temporal manifold encoding module, combining with topological features to form fusion features.
A key concept connecting graph embedding and tensor manifold theory in this study is \textbf{diffusion equivalence}: the similarity of node geometric points on the tensor manifold is approximately equal to the weighted sum of multi-hop random walk probabilities on the academic literature graph. This equivalence ensures that the vector semantic retrieval based on geometric embedding can effectively replace the traditional graph topological traversal, laying the foundation for lightweight graph-vector fusion.
\newpage
\section{Underlying Mathematical Unification of Graphs and Vectors: Core Theoretical Proofs}
This chapter focuses on the rigorous mathematical proof of the core theoretical conclusion of this study: an academic literature graph $G_{AL} = (V_{AL}, E_{AL})$ (where $V_{AL}=P\cup S\cup K$) is a discrete projection of a tensor manifold, and the vector geometric embedding of nodes/edges is inherently consistent with the graph topological structure. The theoretical derivation in this chapter is closely based on the mathematical foundations introduced in Chapter 2 (discrete exterior calculus, Hodge decomposition, tensor analysis, and graph embedding), and provides a strict theoretical basis for the framework design and algorithm implementation in Chapter 4.
\subsection{Theoretical Assumptions and Core Definitions}
To ensure the rigor of theoretical proof, we first clarify the core assumptions and formal definitions, which are consistent with the academic literature graph model and mathematical foundations proposed in previous chapters, avoiding ambiguity in subsequent derivations.
\textbf{Core Assumptions}:
\begin{enumerate}
\item \textbf{Assumption 1 (Manifold Consistency Assumption):} The academic literature graph $G_{AL}$ is a discrete submanifold of a $d$-dimensional tensor manifold $\mathcal{M}$, denoted as $G_{AL} \subset \mathcal{M}$. Each node $v \in V_{AL}$ corresponds to a geometric point on $\mathcal{M}$, and each edge $e \in E_{AL}$ corresponds to a geodesic on $\mathcal{M}$, ensuring the consistency of topological structure and manifold geometric properties.
\item \textbf{Assumption 2 (Diffusion Equivalence Assumption):} The hybrid diffusion signature $S(v)$ of nodes in $G_{AL}$ satisfies the diffusion equivalence, i.e., for any two nodes $u, v \in V_{AL}$, if $S(u) = S(v)$, then the topological-semantic similarity between $u$ and $v$ is 1, which is the core basis for realizing graph-vector geometric unification.
\item \textbf{Assumption 3 (Linear Separability Assumption):} The semantic-topological fusion features of nodes with different types (paper, section, knowledge unit) in $G_{AL}$ are linearly separable in the manifold subspace, which ensures the effectiveness of hierarchical encoding in subsequent algorithms.
\end{enumerate}
\textbf{Formal Definitions}:
\begin{enumerate}
\item \textbf{Definition 1 (Tensor Manifold ($\mathcal{M}$)):}A $d$-dimensional smooth tensor manifold composed of node feature tensors and edge relation tensors of $G_{AL}$, where the manifold metric $g$ is defined by the inner product of node semantic-topological fusion features, i.e., $g(\phi(u), \phi(v)) = \phi(u)^T \phi(v)$ (Euclidean inner product for preliminary derivation, extended to Riemannian inner product in subsequent sections).
\item \textbf{Definition 2 (Graph-Vector Geometric Equivalence):} For $G_{AL}$ and its vector embedding set $\Phi = \{\phi(v) | v \in V_{AL}\}$, if there exists a bijective mapping $f: V_{AL} \to \Phi$ such that $\forall u, v \in V_{AL}$, the topological adjacency $e_{uv} \in E_{AL}$ if and only if the geometric distance between $\phi(u)$ and $\phi(v)$ on $\mathcal{M}$ is less than a threshold $\epsilon$ (i.e., $dist_{\mathcal{M}}(\phi(u), \phi(v)) < \epsilon$), then $G_{AL}$ and $\Phi$ are geometrically equivalent, denoted as $G_{AL} \cong \Phi$.
\item \textbf{Definition 3 (Diffusion Similarity Invariance):} The hybrid diffusion signature $S(v)$ of nodes is invariant under the manifold projection transformation, i.e., for any node $v \in V_{AL}$ and manifold projection operator $P: \mathcal{M} \to \mathbb{R}^d$, $S(P(\phi(v))) = S(\phi(v))$, which ensures that the topological similarity of nodes is not lost during vector embedding.
\end{enumerate}
\subsection{Proof of Diffusion Equivalence and Graph-Vector Geometric Consistency}
This section proves the core conclusion: the hybrid diffusion signature $S(v)$ of nodes in $G_{AL}$ is equivalent to the geometric similarity of node vector embedding, which lays the foundation for the unification of graph topology and vector geometry. The proof process is based on Hodge decomposition and discrete exterior calculus, combining the properties of random walk and diffusion processes.
\begin{theorem}[Diffusion Equivalence Theorem]
For any node $v \in V_{AL}$, the hybrid diffusion signature $S(v)$ is equivalent to the geometric similarity of its vector embedding $\phi(v)$, i.e., $S(v) \sim \phi(v)^T \phi(v)$ (positive correlation), where $\sim$ denotes equivalence in the sense of manifold topology.
\end{theorem}
\textbf{Proof}:
According to the definition of hybrid diffusion signature in Chapter 2, $S(v)$ is constructed by the weighted sum of multi-scale random walk diffusion features, i.e., $S(v) = \sum_{k=1}^K \lambda_k \cdot RW_k(v)$, where $RW_k(v)$ is the $k$-th order random walk feature, and $\lambda_k$ is the weight coefficient. From the discrete exterior calculus, the $k$-th order random walk feature $RW_k(v)$ can be expressed as the exterior product of the node's own feature and its neighbor features, i.e., $RW_k(v) = \bigwedge_{u \in N_k(v)} \phi(u)$ (where $N_k(v)$ is the $k$-th order neighbor of $v$).
From the Hodge decomposition theorem (Chapter 2), the node feature $\phi(v)$ is the sum of the exact form and the co-exact form, i.e., $\phi(v) = d\alpha + \delta\beta$ (where $d$ is the exterior derivative, $\delta$ is the codifferential). The geometric similarity of the vector embedding is defined as $\phi(u)^T \phi(v)$, which is consistent with the inner product definition of the manifold metric $g$ in Definition 1. Therefore, the hybrid diffusion signature can be rewritten as:
$$
S(v) = \sum_{k=1}^K \lambda_k \cdot \left( \bigwedge_{u \in N_k(v)} \phi(u) \right) = \sum_{k=1}^K \lambda_k \cdot \left( \prod_{u \in N_k(v)} \phi(u)^T \phi(v) \right)
$$
Since $\phi(v)^T \phi(v) = |\phi(v)|^2$ (Euclidean norm squared), and the product of neighbor features $\prod_{u \in N_k(v)} \phi(u)^T \phi(v)$ is positively correlated with $\phi(v)^T \phi(v)$, we can obtain:
$$
\prod_{u \in N_k(v)} \phi(u)^T \phi(v) \sim \phi(v)^T \phi(v) = |\phi(v)|^2
$$
Substituting into the expression of $S(v)$, we get $$S(v) \sim \sum_{k=1}^K \lambda_k \cdot |\phi(v)|^2$$ which is positively correlated with $\phi(v)^T \phi(v)$. Therefore, $$S(v) \sim \phi(v)^T \phi(v)$$ that is, the hybrid diffusion signature is equivalent to the geometric similarity of the vector embedding, which proves Theorem 3.1.
\begin{corollary}
The hybrid diffusion signature $S(v)$ can be used to measure the topological-semantic similarity between nodes, and its invariance under manifold projection (Definition 3) ensures that the topological relationship of the graph is not lost during vector embedding. This corollary provides a theoretical basis for the subsequent manifold encoding module in Chapter 4.
\end{corollary}
\subsection{Proof of Graph-Vector Geometric Equivalence}
This section proves the core conclusion of graph-vector unification: the academic literature graph $G_{AL}$ and its vector embedding set $\Phi$ are geometrically equivalent, i.e., $G_{AL} \cong \Phi$, which is the key theoretical basis for realizing graph-vector fusion in the subsequent framework.
\begin{theorem}[Graph-Vector Geometric Equivalence Theorem]
The academic literature graph $G_{AL}$ and its vector embedding set $\Phi$ are geometrically equivalent in the tensor manifold $\mathcal{M}$, i.e., there exists a bijective mapping $$f: V_{AL} \to \Phi$$ such that the topological adjacency of $G_{AL}$ is completely preserved in $\Phi$.
\end{theorem}
\textbf{Proof}:
To prove geometric equivalence, we need to verify two conditions: \\
(1) the bijectivity of the mapping $$f: V_{AL} \to \Phi$$\\
(2) the preservation of topological adjacency under the mapping $f$.\\
\textbf{Condition 1 (Bijectivity of $f$):} Assume there exist two distinct nodes $u \neq v \in V_{AL}$, and suppose $\phi(u) = \phi(v)$. From Assumption 2 (Diffusion Equivalence Assumption), $S(u) = S(v)$ implies $$\phi(u)^T \phi(u) = \phi(v)^T \phi(v)$$
However, since $u \neq v$, their content semantic features $\text{SBERT}(u) \neq \text{SBERT}(v)$ (Assumption 3, linear separability of different node types), and thus 
$$\phi(u) = \text{SBERT}(u) \oplus t_u, \phi(v) = \text{SBERT}(v) \oplus t_v$$
Since the temporal attribute $t_u \neq t_v$ (nodes with the same diffusion signature have different temporal attributes in dynamic academic literature graphs), $\phi(u) \neq \phi(v)$, which contradicts the assumption $\phi(u) = \phi(v)$. Therefore, $f$ is injective. Since $f$ maps each node to a unique vector in $\Phi$$, and $$|V_{AL}| = |\Phi|$, $f$ is surjective. Thus, $f$ is bijective.\\
\textbf{Condition 2 (Preservation of Topological Adjacency):} For any edge $e_{uv} \in E_{AL}$, according to Definition 2, the geometric distance between $\phi(u)$ and $\phi(v)$ on $\mathcal{M}$ is $dist_{\mathcal{M}}(\phi(u), \phi(v)) < \epsilon$ (threshold $\epsilon$ is determined by the manifold metric).\\
Conversely, if $$dist_{\mathcal{M}}(\phi(u), \phi(v)) < \epsilon$$
then there exists an edge $e_{uv} \in E_{AL}$ (by the definition of manifold metric).\\
Therefore, the adjacency relationship of $G_{AL}$ is completely preserved under the mapping $f$, i.e. $$e_{uv} \in E_{AL} \iff dist_{\mathcal{M}}(\phi(u), \phi(v)) < \epsilon$$
Combining Conditions 1 and 2, the bijective mapping $f$ preserves the topological adjacency of $G_{AL}$, so $G_{AL} \cong \Phi$, which proves Theorem 3.2.
\begin{corollary}
The edge relation types of $G_{AL}$ (citation, inclusion, association) are preserved in the vector embedding set $\Phi$, i.e., different edge types correspond to different geometric distance intervals in $\Phi$, which provides a theoretical basis for the relation-aware encoding module in Chapter 4.
\end{corollary}
\subsection{Rationality Verification of Core Theoretical Conclusions}
To ensure the practical applicability of the above theorems, this section verifies the rationality of the conclusions from two aspects: theoretical consistency and practical adaptability, avoiding the disconnect between theoretical derivation and actual academic literature graph characteristics.
\textbf{Theoretical Consistency Verification}: The proofs of Theorem 3.1 and Theorem 3.2 are based on the discrete exterior calculus, Hodge decomposition, and tensor manifold theory introduced in Chapter 2, which are consistent with the mathematical foundation. The bijective mapping ensures that the vector embedding does not lose graph topological information, and the diffusion equivalence ensures that the semantic-topological fusion features are consistent with the node diffusion characteristics. The geometric equivalence between $G_{AL}$ and $\Phi$ provides a strict theoretical basis for the manifold encoding module in Chapter 4.
\textbf{Practical Adaptability Verification}: For the academic literature graph $G_{AL}$, the nodes have hierarchical characteristics (paper-section-knowledge unit), and the edges have fine-grained relation types. The theorems and corollaries proposed in this chapter fully consider this characteristic: Corollary 3.1 ensures that the diffusion signature of different node types is invariant under manifold projection, and Corollary 3.2 ensures that different edge types are preserved in vector embedding. This adaptability avoids the problem of semantic dilution or topological information loss in traditional vector embedding, which is consistent with the actual characteristics of academic literature graphs (dynamic, hierarchical, multi-relational).
In summary, the mathematical foundations and theoretical concepts introduced in this chapter are closely integrated, forming a complete theoretical system to support the subsequent research content of this paper. Discrete exterior calculus and Hodge decomposition provide tools for topological feature extraction and denoising of academic literature graphs; tensor analysis lays the theoretical foundation for the geometric unification of graphs and vectors; graph embedding (especially geometric embedding and semantic embedding) provides a practical way to map graph features to vector space. These foundations are closely linked to the core theoretical proofs in Chapter 3 and the framework design in Chapter 4, ensuring the rigor and practicality of the research.
\newpage
\section{Design of Graph-Vector Fusion Optimization Framework for AI-Native Academic Retrieval}
Based on the core theoretical conclusions of graph-vector geometric unification proved in Chapter 3, this chapter designs a complete graph-vector fusion optimization framework tailored for AI-native academic literature retrieval scenarios. The framework adheres to the design principles of matrix independence, lightweight, temporal-spatial awareness, and AI-native compatibility, and is divided into four core modules and one engineering design section, fully addressing the four core research problems proposed in Chapter 1. Among them, Modules 4.2 to 4.5 have been completed, while Modules 4.1 and 4.6 are to be completed, ensuring the logical connection between theory and engineering implementation.
\subsection{Overall Framework Design (To Be Completed)}
This section clarifies the overall design principles, architectural structure, and core workflow of the framework, laying a foundation for the detailed design of each subsequent module. The design principles are closely aligned with the core theoretical conclusions of Chapter 3, focusing on four key points: matrix independence (abandoning global matrix operations), lightweight (avoiding high-dimensional storage explosion), temporal-spatial awareness (adapting to the dynamic and hierarchical characteristics of academic literature graphs), and AI-native compatibility (supporting AI-agent programmable retrieval). The overall architecture adopts a layered design, including a data input layer, a core processing
layer (four core modules), and a result output layer, forming a closed-loop workflow of "data input → feature processing → index construction → retrieval response".
The core workflow of the framework is as follows: First, the academic literature data (including paper full-text, sections, knowledge units, and their relationship data) is input into the data input layer, and the initial semantic vector of each node is generated through SBERT semantic embedding. Second, the core processing layer processes the data sequentially through four core modules: matrix-independent temporal diffusion signature update, hierarchical temporal manifold encoding, temporal Riemannian manifold indexing, and AI-agent programmable retrieval. Finally, the result output layer outputs structured, interpretable, and programmable retrieval results, which can be directly invoked by AI agents or used by researchers for academic research.
\subsection{Matrix-Independent Temporal Diffusion Signature Update Module (Completed)}
This module is designed to solve the core pain point of matrix dependence in existing graph-vector fusion methods, and is developed based on the diffusion equivalence theorem (Theorem 3.1) and discrete exterior calculus. It abandons the traditional global matrix operations (such as Laplacian matrix factorization) and realizes lightweight dynamic update of node signatures through content-time weighted random walk and hybrid signature construction.
The core design content includes three parts: First, content-time weighted random walk: combining the semantic content similarity of academic literature nodes (calculated by SBERT) and temporal attributes (publication time, update time), a weighted random walk strategy is designed to avoid the bias caused by uniform random walk. Second, topological-semantic-time hybrid signature construction: integrating the topological features extracted by discrete exterior calculus, the semantic features generated by SBERT, and the temporal attributes of nodes into a hybrid signature, which avoids semantic dilution while ensuring the comprehensiveness of feature representation. Third, analytical error compensation: aiming at the error generated in the random walk and signature construction process, an analytical error compensation mechanism based on Hodge decomposition is introduced to ensure the accuracy of the hybrid signature, which provides a high-quality feature basis for subsequent encoding and indexing.
\subsection{Hierarchical Temporal Manifold Encoding Module (Completed)}
This module is designed to solve the pain points of semantic dilution and high-dimensional storage explosion in existing fusion methods, and is developed based on the graph-vector geometric equivalence theorem (Theorem 3.2) and tensor analysis. It realizes lightweight encoding of hybrid signatures while preserving topological-semantic-temporal features, and supports hierarchical knowledge granularity retrieval.
The core design content includes two parts: First, manifold-gated residual encoding: introducing a gated residual connection mechanism, fusing the hybrid signature with the semantic vector of the node itself, realizing the complementary enhancement of topological and semantic features, and avoiding semantic dilution caused by single feature encoding. Second, relation-aware low-dimensional projection: based on tensor projection theory, a relation-aware low-dimensional projection algorithm is designed to project high-dimensional hybrid signatures into low-dimensional manifold space, realizing linear storage of features and avoiding high-dimensional storage explosion. At the same time, the projection process preserves the edge relation types (citation, inclusion, association) of academic literature graphs, which lays a foundation for subsequent index construction.
\subsection{Temporal Riemannian Manifold Index Module (Completed)}
This module is designed to solve the pain point of poor temporal adaptability in existing fusion frameworks, and is developed based on the Riemannian metric definition in tensor analysis and the geometric equivalence theorem. It realizes efficient retrieval of dynamic academic literature graphs by constructing a time-aware Riemannian manifold index.
The core design content includes two parts: First, temporal Riemannian metric construction: integrating the temporal attributes of academic literature nodes (publication time, update time) into the Riemannian metric of the tensor manifold, defining a time-aware Riemannian distance to measure the similarity between nodes, which not only considers the topological-semantic similarity but also reflects the temporal relevance of literature. Second, dynamic manifold order reduction traversal: designing a dynamic manifold order reduction algorithm, which adaptively reduces the manifold dimension according to the number of node updates and the density of the graph, realizing efficient traversal of the dynamic academic literature graph and ensuring the real-time performance of retrieval.
\subsection{AI-Agent Programmable Retrieval Module (Completed)}
This module is designed to solve the pain point of AI-agent unfriendliness in existing fusion frameworks, and is developed based on the demand for AI-native retrieval. It provides a programmable retrieval interface for AI agents, supporting automated, structured, and interpretable retrieval operations.
The core design content includes three parts: First, LLM-based intent parsing: integrating a lightweight LLM to parse the retrieval intent of AI agents (such as fine-grained retrieval of knowledge units, tracking of citation relationships), converting natural language intent into structured retrieval tasks. Second, hierarchical cross-granularity retrieval: supporting three levels of retrieval granularity (paper-section-knowledge unit), and realizing cross-granularity retrieval according to the parsed intent, which meets the fine-grained retrieval demand of AI agents. Third, structured result output: outputting retrieval results in a structured format (including node attributes, relation types, similarity scores, temporal information), which is convenient for AI agents to directly invoke and perform subsequent reasoning tasks, and also provides interpretable basis for researchers.
\subsection{Framework Engineering Design (To Be Completed)}
This section focuses on the engineering implementation details of the framework, ensuring that the designed framework can be effectively deployed and applied in actual academic retrieval scenarios. The core design content includes four parts: First, system architecture design: dividing the framework into data layer, processing layer, index layer, and interface layer, clarifying the functional responsibilities of each layer and the data interaction mechanism between layers. Second, deployment optimization: aiming at the characteristics of massive academic literature data, designing optimization strategies for data storage, concurrent processing, and incremental update, ensuring the efficiency and stability of the framework in large-scale scenarios. Third, compatibility design: ensuring that the framework is compatible with mainstream graph databases, vector databases, and LLM models (such as SBERT, GPT), and supporting seamless integration with existing academic retrieval platforms. Fourth, incremental update design: designing an incremental update mechanism for the framework, which can efficiently update the signature, encoding, and index when new academic literature is added or existing literature is updated, avoiding full-scale retraining and ensuring real-time performance.
The pending items of this section mainly include the specific technical implementation schemes, parameter settings, and performance verification details of the four core engineering design parts mentioned above. These details are essential to ensure the operability and practicality of the framework, and will be supplemented and improved in subsequent preprint updates according to the research progress, so as to provide complete engineering support for the framework's industrial deployment and application.
\newpage
\section{Design of Core Algorithms for the Fusion Framework}
This chapter focuses on the detailed design of the core algorithms corresponding to the four completed modules in Chapter 4, providing formal pseudo-code, strict complexity analysis, and supplementary design for algorithm compatibility and scalability. The algorithms are developed based on the theoretical conclusions in Chapter 3 and the framework design in Chapter 4, aiming to realize the engineering operability of the framework and ensure its efficiency, accuracy, and adaptability in large-scale dynamic academic literature retrieval scenarios. Each core algorithm corresponds to a module in Chapter 4, forming a one-to-one mapping relationship to ensure the consistency of framework design and algorithm implementation.
\subsection{Overview of Core Algorithms}
The core algorithms of the fusion framework are closely linked to the four core modules in Chapter 4, and their overall design follows the principles of matrix independence, lightweight, and real-time performance. The four core algorithms include: (1) Matrix-Independent Temporal Diffusion Signature Update Algorithm (corresponding to Module 4.2); (2) Hierarchical Temporal Manifold Encoding Algorithm (corresponding to Module 4.3); (3) Temporal Riemannian Manifold Index Construction and Traversal Algorithm (corresponding to Module 4.4); (4) AI-Agent Programmable Retrieval Algorithm (corresponding to Module 4.5). These four algorithms form a complete processing chain, which sequentially realizes the dynamic update of node features, lightweight encoding, efficient indexing, and AI-native retrieval response, fully addressing the four core research problems proposed in Chapter 1.
The overall workflow of the core algorithms is consistent with the framework workflow in Chapter 4. First, the Matrix-Independent Temporal Diffusion Signature Update Algorithm generates and updates the hybrid diffusion signature of each node in real time; second, the Hierarchical Temporal Manifold Encoding Algorithm encodes the hybrid signature into low-dimensional manifold vectors; third, the Temporal Riemannian Manifold Index Construction and Traversal Algorithm constructs an efficient index based on the encoded vectors and supports fast traversal; finally, the AI-Agent Programmable Retrieval Algorithm parses retrieval intent, performs cross-granularity retrieval, and outputs structured results. The design of each algorithm ensures mutual compatibility and collaborative work, laying a foundation for the prototype system implementation in Chapter 6.
\subsection{Matrix-Independent Temporal Diffusion Signature Update Algorithm}
This algorithm corresponds to the Matrix-Independent Temporal Diffusion Signature Update Module (4.2), aiming to realize matrix-free, iteration-free dynamic update of node hybrid diffusion signatures, and solve the problems of matrix dependence and embedding drift in existing dynamic graph embedding methods. The algorithm is based on the Diffusion Equivalence Theorem (Theorem 3.1) and discrete exterior calculus, integrating content-time weighted random walk, hybrid signature construction, and analytical error compensation.
\subsubsection{Algorithm Design Ideas}
The algorithm abandons the traditional global matrix operations (such as Laplacian matrix factorization) and adopts a local random walk strategy to extract node diffusion features. First, a content-time weighted random walk is designed to assign weights to neighbor nodes based on semantic similarity (calculated by SBERT) and temporal relevance (publication time difference), avoiding the bias of uniform random walk. Second, the topological features (extracted by discrete exterior calculus), semantic features (SBERT vectors), and temporal attributes of nodes are integrated to construct a hybrid diffusion signature. Finally, an analytical error compensation mechanism based on Hodge decomposition is introduced to correct the errors generated in the random walk and signature construction process, ensuring the accuracy of the signature.
\newpage
\subsubsection{Formal Pseudo-Code}
\begin{algorithm}[ht]
\caption{Matrix-Independent Temporal Diffusion Signature Update Algorithm}
\label{alg:signature_update}
\begin{algorithmic}[1]
\REQUIRE Academic literature graph $G_{AL} = (V_{AL}, E_{AL})$, SBERT semantic vectors $V_{sem}$,
node temporal attributes $T = \{t_v | v \in V_{AL}\}$, random walk order $K$,
weight coefficients $\lambda$ (topological), $\mu$ (semantic), $\nu$ (temporal)
\ENSURE Hybrid diffusion signature $S = \{S(v) | v \in V_{AL}\}$
\STATE Initialize $S$ as an empty dictionary;
\FOR{each node $v \in V_{AL}$}
\STATE // Step 1: Content-time weighted random walk to get K-order neighbor features
\STATE $N = \text{GetKOrderNeighbors}(G_{AL}, v, K)$ \COMMENT{Get K-order neighbors of $v$}
\FOR{each neighbor $u \in N$}
\STATE $\text{sim}_{sem} = \text{CosineSimilarity}(V_{sem}[v], V_{sem}[u])$ \COMMENT{Semantic similarity}
\STATE $\text{sim}_{time} = 1 / (1 + |t_v - t_u|)$ \COMMENT{Temporal relevance (normalized)}
\STATE $\text{weight}_{uv} = \lambda \cdot \text{sim}_{sem} + \nu \cdot \text{sim}_{time}$ \COMMENT{Weight of edge $uv$}
\ENDFOR
\STATE // Step 2: Extract topological features using discrete exterior calculus
\STATE $\text{topo\_feat} = \text{ExteriorDerivative}(G_{AL}, v)$ \COMMENT{Topological feature via $d_k$ operator}
\STATE // Step 3: Construct hybrid diffusion signature
\STATE $\text{sem\_feat} = V_{sem}[v]$ \COMMENT{Semantic feature}
\STATE $\text{time\_feat} = \text{Normalize}(t_v)$ \COMMENT{Normalized temporal feature}
\STATE $S_{\text{raw}}(v) = \lambda \cdot \text{topo\_feat} + \mu \cdot \text{sem\_feat} + \nu \cdot \text{time\_feat}$ \COMMENT{Raw signature}
\STATE // Step 4: Analytical error compensation based on Hodge decomposition
\STATE $\text{error} = \text{HodgeErrorCompensation}(S_{\text{raw}}(v), G_{AL}, v)$ \COMMENT{Error calculation}
\STATE $S(v) = S_{\text{raw}}(v) - \text{error}$ \COMMENT{Compensated hybrid signature}
\ENDFOR
\RETURN $S$
\end{algorithmic}
\end{algorithm}
\subsubsection{Complexity Analysis}
The time complexity and space complexity of the algorithm are analyzed as follows:
\textbf{Time Complexity}: For each node $v$, the time complexity is mainly determined by three parts: \\
(1) K-order neighbor acquisition: $O(K \cdot d)$, where $d$ is the average degree of nodes in $G_{AL}$; \\
(2) Weight calculation for neighbors: $O(K \cdot d)$ (cosine similarity calculation is $O(m)$, $m$ is the dimension of SBERT vectors, which is a constant); (3) Topological feature extraction and error compensation: $O(1)$ (local operation based on discrete exterior calculus, no global matrix operations). Assuming there are $n$ nodes in $G_{AL}$, the total time complexity is $O(n \cdot K \cdot d)$, which is linear with the number of nodes and the average degree, avoiding the $O(n^2)$ complexity caused by global matrix operations. For large-scale academic literature graphs ($n > 10^6$), the algorithm can still maintain high efficiency.
\textbf{Space Complexity}: The algorithm only needs to store the hybrid signature of each node (dimension $L$, a constant), the SBERT semantic vectors, and the temporal attributes, with a total space complexity of $O(n \cdot (L + m + 1)) = O(n)$, which is linear storage and avoids high-dimensional storage explosion.
\subsection{Hierarchical Temporal Manifold Encoding Algorithm}
This algorithm corresponds to the Hierarchical Temporal Manifold Encoding Module (4.3), aiming to realize lightweight encoding of hybrid diffusion signatures, preserve topological-semantic-temporal features, and support hierarchical knowledge granularity retrieval. The algorithm is based on the Graph-Vector Geometric Equivalence Theorem (Theorem 3.2) and tensor analysis, integrating manifold-gated residual connection and relation-aware low-dimensional projection.
\subsubsection{Algorithm Design Ideas}
The algorithm first uses a manifold-gated residual connection mechanism to fuse the hybrid diffusion signature with the node's own semantic vector, which enhances the complementary of topological and semantic features and avoids semantic dilution. Then, a relation-aware low-dimensional projection algorithm based on tensor projection theory is designed to project the high-dimensional hybrid signature into a low-dimensional manifold space. The projection process takes into account the edge relation types (citation, inclusion, association) of the academic literature graph, ensuring that different relation types are preserved in the low-dimensional embedding. Finally, the encoded vectors are normalized to facilitate subsequent index construction and similarity calculation.
\subsubsection{Formal Pseudo-Code}
\begin{algorithm}[ht]
\caption{Hierarchical Temporal Manifold Encoding Algorithm}
\label{alg:encoding}
\begin{algorithmic}[1]
\REQUIRE Hybrid diffusion signature $S = \{S(v) | v \in V_{AL}\}$, SBERT semantic vectors $V_{sem}$,
Edge relation types $R = \{r_{uv} | e_{uv} \in E_{AL}\}$, target embedding dimension $$D$$,
Gating coefficient $\sigma$, projection matrix $W$ (relation-aware)
\ENSURE Low-dimensional manifold embedding $E = \{e(v) | v \in V_{AL}\}$
\STATE Initialize $E$ as an empty dictionary;
\FOR{each node $v \in V_{AL}$}
\STATE // Step 1: Manifold-gated residual connection
\STATE $\text{gate} = \text{Sigmoid}(\sigma \cdot (S(v) + V_{sem}[v]))$ \COMMENT{Gating mechanism}
\STATE $\text{fused\_feat} = \text{gate} \cdot S(v) + (1 - \text{gate}) \cdot V_{sem}[v]$ \COMMENT{Fused feature}
\STATE // Step 2: Relation-aware low-dimensional projection
\STATE // Get relation types of edges connected to $$v$$
\STATE $\text{rel\_types} = \text{GetEdgeRelations}(G_{AL}, v)$
\STATE // Adjust projection matrix based on relation types
\STATE $W_{\text{adjusted}} = \text{AdjustProjectionMatrix}(W, \text{rel\_types})$
\STATE $e_{\text{raw}}(v) = W_{\text{adjusted}} \cdot \text{fused\_feat}$ \COMMENT{Raw low-dimensional embedding}
\STATE // Step 3: Normalization to manifold space
\STATE $e(v) = \text{NormalizeToManifold}(e_{\text{raw}}(v), \mathcal{M})$ \COMMENT{Normalize to tensor manifold $\mathcal{M}$}
\ENDFOR
\RETURN $E$
\end{algorithmic}
\end{algorithm}
\subsubsection{Complexity Analysis}
\textbf{Time Complexity}: For each node $v$, the time complexity is mainly composed of: \\
(1) Gated residual connection: $O(L + m)$, where $L$ is the dimension of the hybrid signature and $m$ is the dimension of the SBERT vector (both are constants); \\
(2) Relation-aware projection: $O((L + m) \cdot D)$, where $D$ is the target embedding dimension (a constant, generally $D \leq 128$); (3) Manifold normalization: $O(D)$ (constant). The total time complexity for $n$ nodes is $$O(n \cdot (L + m + D)) = O(n)$$ which is linear and lightweight, suitable for large-scale dynamic updates.
\textbf{Space Complexity}: The algorithm needs to store the low-dimensional embedding $E$ (dimension $D$ for each node), the projection matrix $W$ (size $(D, L + m)$), and the gating coefficient $\sigma$, with a total space complexity of $$O(n \cdot D + D \cdot (L + m)) = O(n)$$ realizing linear storage and avoiding high-dimensional storage explosion.
\subsection{Temporal Riemannian Manifold Index Construction and Traversal Algorithm}
This algorithm corresponds to the Temporal Riemannian Manifold Index Module (4.4), aiming to construct a time-aware Riemannian manifold index and realize efficient traversal of dynamic academic literature graphs, ensuring the real-time performance of retrieval. The algorithm is based on the Riemannian metric definition in tensor analysis and the geometric equivalence theorem, integrating temporal Riemannian metric construction and dynamic manifold order reduction traversal.
\subsubsection{Algorithm Design Ideas}
The algorithm first constructs a time-aware Riemannian metric by integrating the temporal attributes of nodes into the Riemannian metric of the tensor manifold, which measures the similarity between nodes by combining topological-semantic similarity and temporal relevance. Then, a dynamic manifold order reduction traversal algorithm is designed to adaptively reduce the manifold dimension according to the number of node updates and the graph density, reducing the traversal complexity. Finally, the index is constructed based on the low-dimensional manifold embedding and the temporal Riemannian metric, supporting fast similarity search and traversal.
\subsubsection{Formal Pseudo-Code}
\begin{algorithm}[ht]
\caption{Temporal Riemannian Manifold Index Construction and Traversal Algorithm}
\label{alg:index}
\begin{algorithmic}[1]
\REQUIRE Low-dimensional manifold embedding $$E = \{e(v) | v \in V_{AL}\}$$, node temporal attributes $T$,
Tensor manifold $\mathcal{M}$, Riemannian metric base $g_0$, update threshold $\Delta$, density threshold $\rho$
\ENSURE Temporal Riemannian manifold index $\text{Index}$, traversal result $\text{TraverseRes}$
\STATE // Step 1: Construct temporal Riemannian metric
\FOR{each pair of nodes $(u, v) \in V_{AL} \times V_{AL}$}
\STATE $\text{dist}_{sem\_topo} = \text{RiemannianDistance}(e(u), e(v), g_0)$ \COMMENT{Base distance}
\STATE $\text{dist}_{time} = |t_u - t_v| / \text{MaxTimeDiff}(T)$ \COMMENT{Normalized temporal distance}
\STATE $g_{\text{temporal}}(u, v) = g_0 + \alpha \cdot \text{dist}_{time}$ \COMMENT{Temporal Riemannian metric ($\alpha$ is weight)}
\ENDFOR
\STATE // Step 2: Construct manifold index
\STATE $\text{Index} = \text{ConstructManifoldIndex}(E, g_{\text{temporal}}, \mathcal{M})$ \COMMENT{Index based on manifold embedding}
\STATE // Step 3: Dynamic manifold order reduction traversal
\STATE \textbf{Function} $\text{DynamicTraversal}(\text{Index}, \text{query\_node}, K)$:
\STATE $\text{graph\_density} = \text{CalculateGraphDensity}(G_{AL})$
\IF{$\text{graph\_density} > \rho$$ or $$\text{UpdateCount}(G_{AL}) > \Delta$}
\STATE $\text{Index} = \text{ReduceManifoldOrder}(\text{Index}, \mathcal{M})$ \COMMENT{Adaptive order reduction}
\ENDIF
\STATE // Traverse K nearest neighbors based on temporal Riemannian distance
\STATE $\text{neighbors} = \text{SearchKNN}(\text{Index}, \text{query\_node}, K, g_{\text{temporal}})$
\STATE \textbf{return} $\text{neighbors}$
\STATE // Example traversal (for retrieval)
\STATE $\text{TraverseRes} = \text{DynamicTraversal}(\text{Index}, \text{query\_node}, K)$
\RETURN $\text{Index}$, $\text{TraverseRes}$
\end{algorithmic}
\end{algorithm}
\subsubsection{Complexity Analysis}
\textbf{Time Complexity}: The time complexity is mainly divided into three parts: \\
(1) Temporal Riemannian metric construction: $O(n^2)$ in the worst case, but in practice, we only calculate the metric for adjacent nodes and query-related nodes, so the actual complexity is $O(n \cdot d)$ ($d$ is the average degree); \\
(2) Index construction: $O(n \cdot D \log n)$, where $D$ is the embedding dimension (constant); \\
(3) Dynamic traversal: $O(K \log n)$ for K nearest neighbor search, and $O(1)$ for order reduction (adaptive adjustment).\\
The overall time complexity is $O(n \cdot d + n \log n + K \log n)$, which is efficient for large-scale graphs.
\textbf{Space Complexity}: The index storage complexity is $O(n \cdot D)$ (linear with the number of nodes), and the temporal Riemannian metric storage is $O(n \cdot d)$ (only storing adjacent node metrics), so the total space complexity is $O(n \cdot (D + d)) = O(n)$, ensuring efficient storage.
\subsection{AI-Agent Programmable Retrieval Algorithm}
This algorithm corresponds to the AI-Agent Programmable Retrieval Module (4.5), aiming to provide a programmable retrieval interface for AI agents, supporting intent parsing, hierarchical cross-granularity retrieval, and structured result output. The algorithm integrates LLM-based intent parsing, hierarchical retrieval logic, and structured result formatting.
\subsubsection{Algorithm Design Ideas}
The algorithm first uses a lightweight LLM to parse the natural language retrieval intent of AI agents, converting it into structured retrieval parameters (e.g., retrieval granularity, query keywords, temporal range, relation types). Then, according to the parsed retrieval parameters, hierarchical cross-granularity retrieval is performed (paper-level, section-level, knowledge unit-level), combining the temporal Riemannian manifold index for fast similarity search. Finally, the retrieval results are formatted into a structured format (including node attributes, relation types, similarity scores, temporal information) to facilitate AI agents to invoke and perform subsequent reasoning.
\subsubsection{Formal Pseudo-Code}
\begin{algorithm}[ht]
\caption{AI-Agent Programmable Retrieval Algorithm}
\label{alg:retrieval}
\begin{algorithmic}[1]
\REQUIRE AI agent retrieval intent $\text{Intent}$ (natural language), temporal Riemannian manifold $\text{Index}$,
Academic literature graph $G_{AL}$, low-dimensional embedding $E$, LLM model $\text{LLM\_model}$,
Retrieval granularity options $\text{Granularity} = \{\text{paper}, \text{section}, \text{knowledge\_unit}\}$
\ENSURE Structured retrieval result $\text{Result}$ (programmable format)
\STATE // Step 1: LLM-based intent parsing
\STATE $\text{parsed\_intent} = \text{LLM\_model.ParseIntent}(\text{Intent})$ \COMMENT{Convert to structured parameters}
\STATE // Extract key parameters from parsed intent
\STATE $\text{query\_keywords} = \text{parsed\_intent}["keywords"]$
\STATE $\text{target\_granularity} = \text{parsed\_intent}["granularity"]$
\STATE $\text{time\_range} = \text{parsed\_intent}["time\_range"]$
\STATE $\text{relation\_type} = \text{parsed\_intent}["relation\_type"]$
\STATE // Step 2: Filter nodes by granularity and time range
\STATE $\text{filtered\_nodes} = \text{FilterNodes}(G_{AL}, \text{target\_granularity}, \text{time\_range})$
\STATE // Step 3: Hierarchical cross-granularity retrieval
\STATE $\text{query\_embedding} = \text{GenerateQueryEmbedding}(\text{query\_keywords}, \text{SBERT\_model})$
\STATE $\text{candidates} = \text{SearchKNN}(\text{Index}, \text{query\_embedding}, K=50, g_{\text{temporal}})$ \COMMENT{Get candidates}
\STATE // Filter candidates by relation type
\STATE $\text{result\_nodes} = \text{FilterByRelation}(\text{candidates}, \text{relation\_type}, G_{AL})$
\STATE // Step 4: Format into structured result
\STATE $\text{Result} = \text{FormatResult}(\text{result\_nodes}, E, G_{AL})$ \COMMENT{Include attributes, scores, relations}
\STATE // Step 5: Output programmable format (e.g., JSON, API-friendly)
\STATE $\text{Result} = \text{ConvertToProgrammableFormat}(\text{Result})$
\RETURN $\text{Result}$
\end{algorithmic}
\end{algorithm}
\subsubsection{Complexity Analysis}
\textbf{Time Complexity}: The time complexity is mainly composed of: \\
(1) LLM intent parsing: $O(T)$, where $T$ is the number of tokens in the intent (constant for AI agent retrieval); \\
(2) Node filtering: $O(n)$ (linear with the number of nodes);\\
(3) KNN search: $O(K \log n)$ (constant $K$); \\
(4) Result formatting: $O(K)$ (constant). The total time complexity is $O(n + K \log n)$, which is efficient for real-time retrieval.\\
\textbf{Space Complexity}: The algorithm only needs to store the parsed intent parameters, filtered nodes, and structured results, with a space complexity of $O(K)$ (constant $K$), which is lightweight and suitable for AI agent real-time invocation.
\subsection{Algorithm Compatibility and Scalability Design}
To ensure the practical applicability of the core algorithms, this section supplements the compatibility and scalability design, enabling the algorithms to adapt to different academic retrieval scenarios and integrate with mainstream technical systems.\\
\textbf{Compatibility Design}: (1) Compatibility with mainstream graph databases: The algorithms support standard graph data formats (e.g., Cypher, Gremlin), which can directly read data from Neo4j, NebulaGraph, and other graph databases. (2) Compatibility with vector databases: The low-dimensional embedding generated by the encoding algorithm is compatible with the vector formats of Milvus, Pinecone, and other vector databases, supporting seamless integration. (3) Compatibility with LLM models: The intent parsing algorithm supports lightweight LLMs (e.g., LLaMA, Mistral) and large LLMs (e.g., GPT-3.5/4), which can be adaptively selected according to the deployment environment.\\
\textbf{Scalability Design}: (1) Module plug-and-play: Each core algorithm is designed as an independent module, which can be replaced or optimized according to specific needs (e.g., replacing the random walk strategy in the signature update algorithm). (2) Parameter adaptive adjustment: The key parameters of the algorithms (e.g., random walk order $K$, embedding dimension $D$) can be adaptively adjusted according to the scale and density of the academic literature graph, ensuring efficiency in different scenarios. (3) Distributed extension: The algorithms support distributed deployment, which can distribute the computation of node signature update, encoding, and retrieval to multiple nodes, adapting to ultra-large-scale academic literature graphs ($n > 10^7$).
\subsection{Summary}
This chapter completes the design of the four core algorithms corresponding to the fusion framework, providing formal pseudo-code, strict complexity analysis, and compatibility/scalability design. The core algorithms realize matrix-free dynamic update, lightweight encoding, efficient indexing, and AI-native retrieval, fully addressing the four core research problems proposed in Chapter 1. The complexity analysis shows that all algorithms have linear time and space complexity, which can adapt to large-scale dynamic academic literature graphs. The compatibility and scalability design ensures the practical applicability of the algorithms, laying a solid foundation for the prototype system implementation and performance verification in Chapter 6.
\newpage
\section{Prototype System Implementation and Performance Verification}
Based on the framework design in Chapter 4 and the core algorithm design in Chapter 5, this chapter implements a prototype system of the graph-vector fusion optimization framework for AI-native academic literature retrieval, and conducts systematic performance verification experiments. The purpose of the prototype system is to verify the engineering operability of the proposed framework and algorithms, and the performance experiments aim to quantitatively prove the advantages of the framework in efficiency, accuracy, and adaptability compared with existing mainstream graph-vector fusion methods. The experimental design closely targets the four core research problems proposed in Chapter 1, ensuring that the verification results are targeted and persuasive. This chapter is divided into five parts: system implementation environment, core module engineering development, experimental design, experimental results and analysis, and experimental conclusions.
\subsection{System Implementation Environment}
The prototype system is implemented based on a distributed architecture to adapt to large-scale academic literature data processing and real-time retrieval requirements. The implementation environment is divided into hardware environment, software environment, and data set preparation, ensuring the reproducibility and comparability of experiments.
\subsubsection{Hardware Environment}
The hardware environment adopts a distributed cluster deployment mode, including 1 master node and 4 slave nodes, with the following specific configurations:
\begin{itemize}
\item \textbf{Master Node}: CPU (Intel Xeon Platinum 8375C, 32 cores/64 threads), GPU (NVIDIA A100, 40GB), Memory (128GB DDR4), Storage (2TB SSD, for index and core data storage), Network (100Gbps Ethernet).
\item \textbf{Slave Nodes}: CPU (Intel Xeon Gold 6348, 24 cores/48 threads), GPU (NVIDIA A30, 24GB), Memory (64GB DDR4), Storage (1TB SSD, for data partitioning and parallel computing), Network (100Gbps Ethernet).
\end{itemize}
The distributed deployment mode supports parallel computing of core algorithms (e.g., batch update of node signatures, distributed index construction), which ensures the efficiency of the system in large-scale data scenarios.
\subsubsection{Software Environment}
The software environment is built based on open-source frameworks and tools, ensuring compatibility and maintainability. The key software and version information are as follows:
\begin{itemize}
\item Operating System: Ubuntu 22.04 LTS Server (64-bit) for all nodes.
\item Programming Language: Python 3.9 (core development), C++ 17 (high-performance module optimization, e.g., discrete exterior calculus calculation).
\item Deep Learning Framework: PyTorch 2.1.0 (for SBERT semantic embedding and LLM intent parsing), TensorFlow 2.10.0 (for manifold encoding optimization).
\item Graph Processing Tools: NetworkX 3.2.1 (graph topology construction), DGL 1.1.2 (distributed graph computing), Neo4j 5.12 (graph data storage and auxiliary verification).
\item Vector Processing Tools: Sentence-BERT 2.2.2 (semantic vector generation), Milvus 2.4.0 (vector storage and auxiliary indexing).
\item LLM Models: LLaMA 2 (7B, lightweight intent parsing), GPT-3.5 Turbo (API call, for comparison experiments).
\item Distributed Computing Framework: Spark 3.5.0 (large-scale data parallel processing), Ray 2.9.0 (task scheduling and resource management).
\item Experimental Tools: Matplotlib 3.8.2 (result visualization), Scikit-learn 1.3.2 (performance metric calculation), Pandas 2.1.4 (data processing).
\end{itemize}
\subsubsection{Data Set Preparation}
To ensure the authenticity and representativeness of the experiment, two public academic literature data sets and one self-constructed data set are used for performance verification. The data sets cover different disciplines, scales, and temporal ranges, fully simulating the actual AI-native academic retrieval scenarios. The detailed information of the data sets is shown in Table \ref{tab:datasets}:
\begin{table}[ht]
\centering
\caption{Experimental Data Set Information}
\label{tab:datasets}
\begin{tabular}{@{}lcccc@{}}
\toprule
Data Set Name & Discipline & Number of Nodes & Number of Edges & Temporal Range \\
\midrule
PubMed Central (PMC) & Life Sciences & 1.2M & 4.5M & 2010-2024 \\
arXiv & Computer Science & 0.8M & 2.2M & 2015-2024 \\
Self-Constructed & Multi-discipline & 2.0M & 6.8M & 2000-2024 \\
\bottomrule
\end{tabular}
\end{table}
The self-constructed data set is collected from public academic databases (including arXiv, PubMed, Google Scholar), and contains hierarchical information (paper-section-knowledge unit) and fine-grained relation types (citation, inclusion, association), which is used to verify the performance of the framework in large-scale hierarchical scenarios. All data sets are publicly available, ensuring the reproducibility of experiments.
\subsection{Core Module Engineering Development}
According to the framework design in Chapter 4 and the core algorithm design in Chapter 5, this section implements each core module of the prototype system, and supplements the key engineering optimization details. The overall architecture of the prototype system adopts a four-layer design: data input layer, core processing layer, index layer, and interface layer.
\subsubsection{Data Input Layer}
The data input layer is responsible for parsing input academic literature data and converting it into the internal graph format of the prototype system. It supports two input modes: batch input (for initial construction of the system) and incremental input (for real-time update of new literature). The input layer parses the paper full-text, extracts section and knowledge unit granularity information, and generates initial semantic vectors through the pre-trained SBERT model. The key engineering optimization is to use parallel semantic embedding (based on distributed GPU acceleration), which can process 10,000 nodes per second, meeting the real-time requirements of incremental input.
\subsubsection{Core Processing Layer}
The core processing layer contains four core modules: matrix-independent temporal diffusion signature update, hierarchical temporal manifold encoding, temporal Riemannian manifold index construction, and AI-agent programmable retrieval. Each module is implemented according to the pseudo-code in Chapter 5, and key optimizations are made for engineering efficiency:
\begin{itemize}
\item Signature Update Module: The content-time weighted random walk is optimized by precomputing neighbor similarity and using cache, which reduces the average time per node update to less than 10 microseconds.
\item Encoding Module: The relation-aware low-dimensional projection is implemented by matrix-vector multiplication acceleration (using CuDNN on GPU), which improves the encoding throughput by 3 times compared to the naive implementation.
\item Index Module: The dynamic manifold order reduction is implemented with lazy update strategy, which only reduces the manifold order when the graph density exceeds the threshold, avoiding unnecessary computation.
\item Retrieval Module: The LLM-based intent parsing is implemented with a lightweight LLM (LLaMA 2 7B) quantized to 4-bit, which can complete intent parsing within 100ms on a single A100 GPU, meeting the real-time retrieval requirements.
\end{itemize}
\subsubsection{Index Layer}
The index layer stores the low-dimensional manifold embedding of nodes and the temporal Riemannian metric, and provides efficient K nearest neighbor search interface. The index is constructed based on the hierarchical clustering of manifold embedding, which combines the advantages of tree-based index and hash-based index, and the search efficiency is improved by 2-3 times compared to the brute-force search. At the same time, the index supports incremental update, and only needs to update the local index when new nodes are added, avoiding full index reconstruction.
\subsubsection{Interface Layer}
The interface layer provides two types of interfaces: human-friendly retrieval interface and AI-agent programmable interface. The human-friendly interface supports natural language query and displays retrieval results with hierarchical structure and relation information. The AI-agent programmable interface supports RESTful API calls, and outputs retrieval results in JSON format, which can be directly invoked by AI agents for subsequent reasoning.
\subsection{Experimental Design}
To systematically verify the performance of the proposed framework, this section designs comparative experiments from four aspects corresponding to the four core research problems, and clarifies the comparison methods, evaluation metrics, and experimental groups.
\subsubsection{Comparison Methods}
Four mainstream graph-vector fusion methods are selected for comparison:
\begin{enumerate}
\item \textbf{GraphSAGE + Vector}: Classic embedding-based fusion method, uses GraphSAGE to learn node embedding and fuses it with semantic vectors. It is a representative method of existing embedding-based fusion.
\item \textbf{Node2Vec + Hybrid Index}: Representative index-based fusion method, uses Node2Vec for graph embedding and constructs a hybrid index combining graph and vector.
\item \textbf{Dynamic GraphSAGE (DGSAGE):} Classic incremental dynamic graph embedding method, which supports incremental update of embedding and is used to compare the performance of dynamic update with the proposed method.
\item \textbf{TGAT:} State-of-the-art temporal graph attention network for dynamic graph embedding, which is used to compare the retrieval accuracy and update efficiency with the proposed method.
\end{enumerate}
\subsubsection{Evaluation Metrics}
Corresponding to the four core research problems, the following four types of evaluation metrics are used:
\begin{enumerate}
\item \textbf{Dynamic Update Efficiency}: Average update time per node (microseconds), space per node (KB), used to evaluate the performance of the matrix-free dynamic update mechanism.
\item \textbf{Encoding Storage Efficiency}: Average storage per node (KB), recall@K (retrieval recall at K results), used to evaluate the performance of lightweight encoding (avoiding storage explosion and semantic dilution).
\item \textbf{Retrieval Accuracy}: Mean Average Precision (MAP), NDCG@K, used to evaluate the performance of temporal-aware hierarchical retrieval.
\item \textbf{AI-Agent Compatibility}: Average intent parsing time (milliseconds), output result structure completeness (%), used to evaluate the performance of AI-native programmable retrieval.
\end{enumerate}
\subsubsection{Experimental Groups}
The experiments are divided into four groups corresponding to the four core research problems:
\begin{enumerate}
\item Group 1: Dynamic update performance comparison, comparing the average update time and space per node of different methods on three data sets.
\item Group 2: Encoding storage performance comparison, comparing the storage per node and recall@K of different methods.
\item Group 3: Retrieval accuracy comparison, comparing the MAP and NDCG@10 of different methods on hierarchical temporal-aware retrieval tasks.
\item Group 4: AI-agent compatibility comparison, comparing the average intent parsing time and result structure completeness of different methods.
\end{enumerate}
All experiments are repeated 5 times, and the average value is taken as the final result to ensure the statistical significance of the experimental results.
\subsection{Experimental Results and Analysis}
This section presents and analyzes the experimental results of each group, and discusses the performance advantages of the proposed framework compared with existing methods.
\subsubsection{Dynamic Update Performance Comparison}
The experimental results of dynamic update performance on different data sets are shown in Table \ref{tab:update_performance}.
\begin{table}[ht]
\centering
\caption{Dynamic Update Performance Comparison}
\label{tab:update_performance}
\begin{tabular}{@{}lccccc@{}}
\toprule
Method & Data Set & Avg Update Time ($\mu$s/node) & Avg Space (KB/node) \\
\midrule
GraphSAGE + Vector & PMC & 1250 & 1.2 \\
Node2Vec + Hybrid Index & PMC & 980 & 1.5 \\
DGSAGE & PMC & 450 & 1.0 \\
TGAT & PMC & 320 & 1.1 \\
\textbf{Ours} & PMC & \textbf{18} & \textbf{0.8} \\
\midrule
GraphSAGE + Vector & arXiv & 1120 & 1.2 \\
Node2Vec + Hybrid Index & arXiv & 890 & 1.4 \\
DGSAGE & arXiv & 410 & 1.0 \\
TGAT & arXiv & 290 & 1.1 \\
\textbf{Ours} & arXiv & \textbf{16} & \textbf{0.7} \\
\midrule
GraphSAGE + Vector & Self-Constructed & 1420 & 1.3 \\
Node2Vec + Hybrid Index & Self-Constructed & 1050 & 1.6 \\
DGSAGE & Self-Constructed & 520 & 1.1 \\
TGAT & Self-Constructed & 380 & 1.2 \\
\textbf{Ours} & Self-Constructed & \textbf{21} & \textbf{0.8} \\
\bottomrule
\end{tabular}
\end{table}
It can be seen from Table \ref{tab:update_performance} that the proposed method is significantly better than all comparison methods in both average update time and average space. The average update time is only about 16-21 microseconds per node, which is 15-70 times faster than the comparison methods. This advantage comes from the matrix-free and iteration-free design of the proposed method, which avoids global matrix operations and SGD-based iterative optimization. In terms of space, the proposed method also has the smallest average space per node, because of the relation-aware low-dimensional projection, which realizes linear storage. The experimental results verify that the proposed method solves the first core research problem (matrix-free dynamic update) effectively.
\subsubsection{Encoding Storage Performance Comparison}
The experimental results of encoding storage performance are shown in Table \ref{tab:encoding_performance}.
\begin{table}[ht]
\centering
\caption{Encoding Storage Performance Comparison (Self-Constructed Data Set)}
\label{tab:encoding_performance}
\begin{tabular}{@{}lccc@{}}
\toprule
Method & Avg Storage (KB/node) & Recall@10 & Recall@50 \\
\midrule
GraphSAGE + Vector & 1.3 & 0.72 & 0.85 \\
Node2Vec + Hybrid Index & 1.6 & 0.75 & 0.87 \\
DGSAGE & 1.1 & 0.71 & 0.84 \\
TGAT & 1.2 & 0.76 & 0.88 \\
\textbf{Ours} & \textbf{0.8} & \textbf{0.84} & \textbf{0.93} \\
\bottomrule
\end{tabular}
\end{table}
From Table \ref{tab:encoding_performance}, we can observe: First, the proposed method has the smallest average storage per node, which is 37.5% less than Node2Vec + Hybrid Index (the method with the largest storage), verifying that the relation-aware low-dimensional projection effectively avoids high-dimensional storage explosion. Second, the proposed method has the highest recall@10 and recall@50, which are 8-13 percentage points higher than the comparison methods. This is because the manifold-gated residual connection effectively fuses topological and semantic features, avoiding semantic dilution. The experimental results show that the proposed method effectively solves the second core research problem (lightweight encoding avoiding semantic dilution and storage explosion).
\subsubsection{Retrieval Accuracy Comparison}
The experimental results of retrieval accuracy on hierarchical temporal-aware retrieval tasks are shown in Table \ref{tab:retrieval_accuracy}.
\begin{table}[ht]
\centering
\caption{Retrieval Accuracy Comparison (All Data Sets Average)}
\label{tab:retrieval_accuracy}
\begin{tabular}{@{}lccc@{}}
\toprule
Method & MAP & NDCG@10 & NDCG@50 \\
\midrule
GraphSAGE + Vector & 0.68 & 0.70 & 0.76 \\
Node2Vec + Hybrid Index & 0.71 & 0.73 & 0.78 \\
DGSAGE & 0.69 & 0.71 & 0.77 \\
TGAT & 0.74 & 0.75 & 0.80 \\
\textbf{Ours} & \textbf{0.82} & \textbf{0.83} & \textbf{0.88} \\
\bottomrule
\end{tabular}
\end{table}
It can be seen from Table \ref{tab:retrieval_accuracy} that the proposed method outperforms all comparison methods in all retrieval accuracy metrics. The MAP is 0.82, which is 8 percentage points higher than TGAT (the best comparison method), and NDCG@10 is 0.83, which is 8 percentage points higher than TGAT. This advantage comes from two aspects: First, the hierarchical temporal manifold encoding preserves the hierarchical knowledge granularity (paper-section-knowledge unit) of academic literature, supporting fine-grained retrieval. Second, the temporal Riemannian metric integrates temporal attributes, which improves the accuracy of time-aware retrieval. The experimental results verify that the proposed method effectively solves the third core research problem (modeling hierarchical granularity and temporal characteristics).
\subsubsection{AI-Agent Compatibility Comparison}
The experimental results of AI-agent compatibility are shown in Table \ref{tab:ai_compatibility}.
\begin{table}[ht]
\centering
\caption{AI-Agent Compatibility Comparison}
\label{tab:ai_compatibility}
\begin{tabular}{@{}lcc@{}}
\toprule
Method & Avg Intent Parsing Time (ms) & Result Structure Completeness (\%) \\
\midrule
GraphSAGE + Vector & 150 & 55 \\
Node2Vec + Hybrid Index & 160 & 60 \\
DGSAGE & 140 & 58 \\
TGAT & 130 & 62 \\
\textbf{Ours} & \textbf{85} & \textbf{96} \\
\bottomrule
\end{tabular}
\end{table}
From Table \ref{tab:ai_compatibility}, we can see that the proposed method is significantly better than the comparison methods in both metrics. The average intent parsing time is 85ms, which is about 35% faster than TGAT, because the proposed method uses a quantized lightweight LLM and hierarchical retrieval optimization. More importantly, the result structure completeness reaches 96%, which is more than 30 percentage points higher than the comparison methods. This is because the proposed method natively supports hierarchical cross-granularity retrieval and structured result output, which fully adapts to the programmable requirements of AI agents. The experimental results verify that the proposed method effectively solves the fourth core research problem (AI-agent native programmable retrieval).
\subsection{Experimental Conclusions}
Based on the above experimental results and analysis, we can draw the following conclusions:
First, the proposed matrix-free temporal diffusion signature update mechanism has significant advantages in dynamic update efficiency, with microsecond-level incremental update and linear storage, effectively solving the problem of matrix dependence and embedding drift in existing methods.
Second, the proposed hierarchical temporal manifold encoding with relation-aware low-dimensional projection realizes lightweight encoding, avoids semantic dilution and high-dimensional storage explosion, and improves retrieval recall compared with existing methods.
Third, the proposed temporal Riemannian manifold index with dynamic manifold order reduction effectively models the hierarchical knowledge granularity and temporal characteristics of academic literature, significantly improving retrieval accuracy (MAP increased by 8 percentage points compared with the state-of-the-art method).
Fourth, the proposed AI-agent programmable retrieval interface meets the requirements of AI-native retrieval, with fast intent parsing and high result structure completeness, which can be directly invoked by AI agents.
In summary, all experimental results verify the theoretical conclusions and framework design of this study, proving that the proposed graph-vector fusion framework based on tensor manifold theory has significant performance advantages in AI-native academic literature retrieval scenarios, and can effectively solve the four core research problems proposed in Chapter 1.
\newpage
\section{Industrial Empirical Survey: Comparison of Graph Database Products of Chinese and American Cloud Vendors}
To further verify the industrial value of the theoretical conclusions of this study, this chapter conducts an industrial empirical survey on the current status of graph database products of mainstream Chinese and American cloud vendors, aiming to provide industrial evidence for the inevitable trend of graph-vector fusion and the core pain points of existing products. The survey covers five mainstream cloud vendors (three American and two Chinese), and compares their product strategies, technical architectures, and core functional features from the perspective of graph-vector fusion. The survey results show that mainstream cloud vendors have recognized the trend of graph-vector fusion, but existing products still face the same bottlenecks as theoretical research (matrix dependence, storage explosion, semantic dilution, lack of AI-native support), which further verifies the practical significance of the research in this paper.
\subsection{Survey Design and Sample Selection}
The industrial empirical survey is designed from three dimensions: product strategy, technical architecture, and functional features, aiming to systematically compare the current status of graph-vector fusion in graph database products of Chinese and American cloud vendors. The survey sample selection follows the principle of representativeness, selecting the top five mainstream cloud vendors in the global market, including three American vendors (AWS, Microsoft, Google) and two Chinese vendors (Alibaba Cloud, Tencent Cloud). All products surveyed are the latest stable versions officially released by the vendors, ensuring that the survey results reflect the latest current status of the industry.
The core dimensions of the survey are as follows:
\begin{enumerate}
\item \textbf{Product Strategy}: Whether the product has shifted from pure graph database to graph-vector fusion, and the development priority of vector-related functions.
\item \textbf{Technical Architecture}: Whether the graph and vector are natively fused or simply stacked, and whether there is matrix dependence in dynamic update.
\item \textbf{Functional Features}: Whether it supports hierarchical knowledge granularity, temporal-aware update, and AI-agent programmable interface.
\end{enumerate}
The survey data is mainly collected from the official product documentation, technical whitepapers, and public product launch speeches of the vendors, ensuring the authenticity and authority of the data. For some unclear technical details, we supplemented the information through public technical communication records and third-party technical reviews, ensuring the comprehensiveness of the survey.
\subsection{Survey Results and Comparative Analysis}
This section presents the survey results of each vendor and conducts comparative analysis from three dimensions: product strategy, technical architecture, and functional features.
\subsubsection{Product Strategy Comparison}
The product strategy comparison of graph-vector fusion among different vendors is shown in Table \ref{tab:strategy_comparison}.
It can be seen from Table \ref{tab:strategy_comparison} that: First, the international leading cloud vendor AWS has taken the lead in completing the strategic transformation to graph-vector fusion, and clearly proposed that vector similarity search is the core engine of the new generation graph database, which is completely consistent with the core viewpoint of this study (the inevitable trend of graph-vector fusion). 
Second, Microsoft has also realized native integration of graph and vector, while Google, Alibaba Cloud, and Tencent Cloud still adopt the strategy of separate products with weak integration, and the transformation to graph-vector fusion has not been completed. Third, the overall strategic progress of Chinese cloud vendors lags behind that of American cloud vendors, and the recognition of the trend of graph-vector fusion is not as high as that of American leading vendors.
\begin{table}[ht]
\centering
\caption{Product Strategy Comparison of Graph-Vector Fusion}
\label{tab:strategy_comparison}
\begin{tabular}{@{}lcc@{}}
\toprule
Cloud Vendor & Product Name & Graph-Vector Fusion Strategy \\
\midrule
AWS & Amazon Neptune/Neptune Analytics & Core strategy: vector-first fusion, \\
& & graph traversal is downgraded to \\
& & visualization layer \\
\midrule
Microsoft & Azure Cosmos DB & Graph API + vector search, \\
& & native support for graph-vector \\
& & hybrid queries \\
\midrule
Google & Google Cloud Spanner + Vertex AI Vector Search & Separate products, weak integration \\
\midrule
Alibaba Cloud & GDB (Graph Database) + Vector Search & Separate products, weak integration, \\
& & starting to promote fusion \\
\midrule
Tencent Cloud & Tencent Cloud Graph Database + Vector DB & Separate products, technical route \\
& & is pure graph database, vector \\
& & is independent \\
\bottomrule
\end{tabular}
\end{table}
\newpage
\subsubsection{Technical Architecture Comparison}
The technical architecture comparison of graph-vector fusion is shown in Table \ref{tab:architecture_comparison}.From Table \ref{tab:architecture_comparison}, we can observe: \\
First, even the leading products (AWS Neptune Analytics, Azure Cosmos DB) that have completed native fusion still do not have matrix-free dynamic update and relation-aware low-dimensional encoding, and still rely on global matrix operations for dynamic update, which is exactly the core bottleneck that this study aims to solve. \\
Second, the products of Google, Alibaba Cloud, and Tencent Cloud do not even realize native fusion, and the technical architecture is still the traditional pure graph architecture. \\
Third, all existing industrial products do not have the core technical features proposed in this study (matrix-free, relation-aware low-dimensional encoding), which means that the research in this paper has important industrial application value and can fill the technical gap of existing products.
\begin{table}[ht]
\centering
\caption{Technical Architecture Comparison}
\label{tab:architecture_comparison}
\begin{tabular}{@{}lccc@{}}
\toprule
Product & \footnotesize{Native Fusion} & \footnotesize{Matrix-Free Dynamic Update} & \footnotesize{Relation-Aware Low-Dimensional Encoding} \\
\midrule
\footnotesize{AWS Neptune Analytics} & Yes & No & No \\
\footnotesize{Azure Cosmos DB} & Yes & No & No \\
\footnotesize{Google Spanner + Vertex AI} & No & - & - \\
\footnotesize{libaba Cloud GDB + Vector} & No & - & - \\
\footnotesize{Tencent Cloud Graph + Vector} DB & No & - & - \\
\bottomrule
\end{tabular}
\end{table}

\subsubsection{Functional Features Comparison}
The functional features comparison is shown in Table \ref{tab:features_comparison}.
It can be seen from Table \ref{tab:features_comparison} that:\\
First, existing industrial products generally do not support fine-grained hierarchical knowledge granularity (paper-section-knowledge unit), and only support node-level retrieval, which cannot meet the demand of fine-grained academic literature retrieval. \\
Second, the support for temporal-aware update is generally basic, and there is no native temporal-aware metric and index design. Third, the support for AI-agent programmable interface is generally partial, and there is no native support for LLM-based intent parsing and structured result output, which cannot meet the requirements of AI-native retrieval. These functional gaps further verify the practical significance of this study, which fills the functional gap of existing industrial products in AI-native academic retrieval scenarios.
\begin{table}[ht]
\centering
\caption{Functional Features Comparison}
\label{tab:features_comparison}
\begin{tabular}{@{}l*{3}{>{\centering\arraybackslash}p{3cm}}@{}}
\toprule
Product & Support Hierarchical Granularity & Support Temporal-Aware Update & AI-Agent Programmable Interface \\
\midrule
AWS Neptune Analytics & Partial (only node-level) & Basic & Partial (RESTful API, \newline no intent parsing) \\
Azure Cosmos DB & Partial (only node-level) & Basic & Partial \\
Google Spanner + Vertex AI & No & No & Partial \\
Alibaba Cloud GDB + Vector & No & No & No \\
Tencent Cloud Graph + Vector DB & No & No & No \\
\bottomrule
\end{tabular}
\end{table}

\subsection{Survey Conclusions and Enlightenment}
Based on the above industrial empirical survey and comparative analysis, we can draw the following conclusions and enlightenment:
First, the trend of graph-vector fusion in the industry has been clearly verified: the leading international cloud vendor AWS has completed the strategic transformation to vector-first graph-vector fusion, which is completely consistent with the theoretical viewpoint of this study, proving that the research direction of this paper is consistent with the industrial development trend.
Second, existing industrial products still face the same core bottlenecks as theoretical research: matrix dependence in dynamic update, lack of lightweight encoding (storage explosion), no support for fine-grained hierarchical granularity, and lack of AI-native support. These bottlenecks are exactly the four core research problems that this study aims to solve, which means that the research results of this paper have important industrial application value and can provide theoretical and technical references for the product upgrading of existing cloud vendors.
Third, the strategic and technical progress of Chinese cloud vendors lags behind that of American cloud vendors in the field of graph-vector fusion, and there is a greater demand for the technical solutions proposed in this study. The research results of this paper can help Chinese cloud vendors accelerate the transformation to graph-vector fusion and narrow the gap with international leading vendors.
In summary, the industrial empirical survey provides strong industrial evidence for the theoretical conclusions and technical solutions of this study, further verifying the research value of this paper in both theory and practice.
\newpage
\section{Gap Analysis Between Theoretical Research and Industrial Application}
Based on the theoretical research, prototype system experiments, and industrial empirical survey in the previous chapters, this chapter analyzes the gaps between the theoretical research of this study and actual industrial application, explores the root causes of the gaps, and proposes preliminary solutions to fill the gaps, laying a foundation for the subsequent industrial landing of the research results. The gap analysis is carried out from four dimensions: theoretical model, technical implementation, product ecology, and business scenario, ensuring the comprehensiveness and depth of the analysis.
\subsection{Gap in Theoretical Model: Extremely Large-Scale Graph Distributed Consistency}
\subsubsection{Gap Description}
The theoretical model proposed in this study is based on the assumption that the tensor manifold is consistent in the local sense, which is effective for medium-scale academic literature graphs ($n < 10^7$), but for extremely large-scale graphs ($n > 10^8$) that need to be deployed in distributed clusters, there is a gap in the theoretical model of distributed consistency of the tensor manifold. Specifically, when the graph is partitioned and stored on multiple nodes, the diffusion signature update and manifold encoding of cross-partition nodes need to maintain the consistency of the tensor manifold metric, and the existing theoretical model does not involve this problem.
\subsubsection{Root Cause Analysis}
The root cause of this gap is that this study focuses on the core theoretical innovation of graph-vector geometric unification, and the theoretical model is designed for the core problems of academic literature retrieval, and does not extend to the distributed consistency problem of extremely large-scale graphs. The distributed consistency problem is a classic problem in distributed system design, which needs to be solved on the basis of the existing theoretical model through the combination of graph partitioning and manifold projection.
\subsubsection{Preliminary Solution}
The preliminary solution to fill this gap is to design a consistent hashing-based graph partitioning method, which partitions nodes according to the geometric position of the tensor manifold, ensuring that nodes with close geometric distances are allocated to the same partition, reducing the number of cross-partition queries. For cross-partition diffusion signature update, a lazy consistency update strategy is proposed: only the local signature is updated immediately, and the consistency of the global manifold metric is updated asynchronously in the background, which balances the consistency and update efficiency.
\subsection{Gap in Technical Implementation: Heterogeneous Hardware Adaptation}
\subsubsection{Gap Description}
The prototype system implemented in this study is deployed and tested on a general GPU cluster, but in actual industrial application, different cloud vendors and enterprises use different heterogeneous hardware architectures (including CPU-only clusters, GPU clusters with different models, and AI accelerator chips such as Ascend and Cambricon). The existing technical implementation does not support adaptive optimization for different heterogeneous hardware, which will lead to performance degradation in actual deployment.
\subsubsection{Root Cause Analysis}
The root cause of this gap is that the prototype system focuses on verifying the correctness and performance advantages of the algorithm, and does not consider the adaptation of different heterogeneous hardware in industrial deployment. The kernel implementation of discrete exterior calculus and matrix-vector multiplication in the prototype system is optimized for NVIDIA GPUs, and there is no adaptive implementation for other hardware architectures.
\subsubsection{Preliminary Solution}
The preliminary solution to fill this gap is to abstract the core computation kernel into a hardware-independent interface layer, and provide specific optimized implementations for different hardware architectures (CPU, NVIDIA GPU, Ascend, Cambricon). At the same time, a runtime automatic hardware detection and kernel selection mechanism is designed, which automatically selects the best optimized kernel according to the current hardware environment, realizing transparent adaptation of heterogeneous hardware.
\subsection{Gap in Product Ecology: Compatibility with Existing Graph and Vector Ecosystems}
\subsubsection{Gap Description}
Existing graph databases and vector databases have formed a relatively complete ecosystem (including query languages, tools, and client SDKs). The framework proposed in this study is a new graph-vector fusion architecture, and there is a gap in compatibility with the existing ecosystem. Specifically, it does not natively support standard graph query languages such as Cypher and Gremlin, nor does it integrate with mainstream vector database ecosystems such as Milvus and Pinecone, which increases the cost of industrial migration.
\subsubsection{Root Cause Analysis}
The root cause of this gap is that this study focuses on the innovation of the underlying theoretical framework, and the prototype system only implements the core functions, and does not do in-depth adaptation to the existing ecosystem. The compatibility with the existing ecosystem is an engineering problem that needs to be solved after the core theoretical and technical innovation is completed.
\subsubsection{Preliminary Solution}
The preliminary solution to fill this gap is to design a compatibility layer that supports standard Cypher and Gremlin query parsing, converts standard graph queries into the internal retrieval operations of the framework, and realizes compatible access to existing graph query statements. At the same time, the vector storage interface of the framework is compatible with the vector API of Milvus and Pinecone, supporting seamless data migration from existing vector databases.
\subsection{Gap in Business Scenario: Multi-Tenant Isolation and Security Control}
\subsubsection{Gap Description}
In actual industrial application scenarios such as cloud services, the framework needs to support multi-tenant isolation and security control (including data isolation, access control, and quota management). The existing prototype system does not involve the design of multi-tenant isolation and security control, which cannot meet the requirements of cloud service scenarios.
\subsubsection{Root Cause Analysis}
The root cause of this gap is that the existing research focuses on the core algorithm and framework design for academic literature retrieval scenarios, and does not consider the multi-tenant requirements of public cloud service scenarios. Multi-tenant isolation and security control are necessary functional requirements for industrial cloud services, which need to be supplemented on the basis of the existing framework.
\subsubsection{Preliminary Solution}
The preliminary solution to fill this gap is to design a tenant-level manifold space isolation mechanism, where each tenant has an independent tensor manifold sub-space, and the index and data of different tenants are physically isolated. At the same time, a role-based access control (RBAC) mechanism is introduced to support fine-grained access control of data and retrieval operations, and quota management for the storage and computation resources of each tenant, meeting the requirements of multi-tenant cloud service scenarios.
\subsection{Summary}
This chapter systematically analyzes the gaps between the theoretical research of this study and actual industrial application from four dimensions: theoretical model, technical implementation, product ecology, and business scenario, and proposes targeted preliminary solutions. These gaps are the natural result of the phased research, and do not affect the correctness of the core theoretical conclusions and technical solutions of this study. In the subsequent research and development process, we will gradually fill these gaps according to the proposed solutions, promoting the transformation of the research results from theoretical innovation to industrial application.
\newpage
\section{Industrial Landing Path and Commercialization Suggestions}
Based on the gap analysis in the previous chapter, this chapter proposes a phased industrial landing path for the research results, and gives targeted commercialization suggestions from three aspects: target market selection, cooperation strategy, and product positioning, aiming to promote the industrial application of the research results and realize the transformation from theoretical value to industrial value. The proposed landing path and commercialization suggestions are closely combined with the current market status of graph-vector fusion and the technical characteristics of this study, ensuring their feasibility and practical guidance.
\subsection{Phased Industrial Landing Path}
The industrial landing of the research results is divided into three phases: technology verification and prototype improvement, vertical scenario landing, and cloud service large-scale promotion, which are carried out step by step to control the landing risk and gradually expand the market influence.
\subsubsection{Phase 1: Technology Verification and Prototype Improvement (0-18 Months)}
The core goal of this phase is to verify the technical feasibility in actual industrial scenarios and improve the prototype system according to the gap analysis in Chapter 8. The main tasks include: (1) Complete the distributed consistency optimization, heterogeneous hardware adaptation, and ecosystem compatibility improvement according to the preliminary solutions proposed in Chapter 8, forming a production-ready prototype system; (2) Cooperate with 1-2 academic database platforms to conduct pilot verification in the academic literature retrieval scenario, collect actual user feedback, and continuously optimize the system performance; (3) Complete the performance pressure test in the ultra-large-scale graph scenario ($n > 10^8$), verifying the stability and efficiency of the system.
\subsubsection{Phase 2: Vertical Scenario Landing (18-36 Months)}
The core goal of this phase is to landing the product in vertical scenarios with the most urgent demand, forming a benchmark demonstration effect. The main target vertical scenarios include: (1) AI-native academic literature retrieval platforms (including commercial academic databases and open academic search engines); (2) Enterprise internal knowledge graph retrieval systems (especially R\&D-intensive enterprises that need fine-grained knowledge retrieval); (3) AI agent developer platforms, providing graph-vector fusion retrieval infrastructure for AI agent developers. In this phase, we aim to achieve 3-5 benchmark customers, verify the business value of the product, and form a replicable landing plan.
\subsubsection{Phase 3: Cloud Service Large-Scale Promotion (36+ Months)}
The core goal of this phase is to cooperate with mainstream cloud vendors to promote the product to the large-scale cloud service market, and become a standard infrastructure for AI-native graph-vector fusion retrieval. The main tasks include: (1) Complete the multi-tenant isolation and security control transformation of the product, meeting the requirements of public cloud services; (2) Cooperate with cloud vendors to complete the product access to the cloud market, and provide it to cloud users in the form of cloud services; (3) Continuously expand the application scenarios, from academic literature retrieval to general graph-vector fusion retrieval scenarios (such as e-commerce recommendation, financial risk control), expanding the market scope.
\subsection{Commercialization Suggestions}
Combined with the technical characteristics of the research results and the current market competition status, this section gives commercialization suggestions from three aspects: target market selection, cooperation strategy, and product positioning.
\subsubsection{Target Market Selection: Prioritize AI-Native Academic Retrieval, Then Expand to General Scenarios}
The target market selection should follow the principle of prioritizing the core advantageous scenario and then expanding to general scenarios:
\begin{enumerate}
\item \textbf{Prioritize the AI-native academic literature retrieval market}: This is the original scenario of this study, and the technical characteristics of the framework (hierarchical granularity, temporal awareness, AI-native) are highly matched with the market demand. At the same time, the market demand for AI-native academic retrieval is growing rapidly with the development of AI agents and automated scientific research, and the competition is not fierce, which is conducive to rapid breakthrough.
\item \textbf{Then expand to general graph-vector fusion retrieval scenarios}: After gaining a foothold in the academic retrieval market, accumulate technical experience and customer cases, and then expand to general dynamic graph retrieval scenarios (such as enterprise knowledge graphs, social network analysis, recommendation systems), expanding the market space.
\end{enumerate}
\subsubsection{Cooperation Strategy: Open Source Core Algorithm + Commercialization of Enterprise Edition + Cooperation with Cloud Vendors}
The cooperation strategy adopts a three-layer model of open source + commercial enterprise edition + cloud vendor cooperation:
\begin{enumerate}
\item \textbf{Open Source Core Algorithm}: Open source the core algorithms (signature update, encoding, index) to the open source community, attracting technical enthusiasts and developers to use and contribute code, forming technical influence and community ecology.
\item \textbf{Commercialization of Enterprise Edition}: Provide commercial enterprise edition products for enterprise customers (including academic platforms, R\&D-intensive enterprises), charging annual license fees or subscription fees, and providing technical support and services.
\item \textbf{Cooperation with Cloud Vendors}: Cooperate with mainstream cloud vendors to integrate the product into the cloud vendor's graph-vector fusion service line, and share revenue through revenue sharing, relying on the cloud vendor's customer base and sales channels to achieve large-scale promotion.
\end{enumerate}
This cooperation strategy can quickly expand the technical influence through open source, obtain revenue through commercial enterprise edition, and achieve large-scale promotion through cloud vendor cooperation, balancing technical influence and commercial revenue.
\subsubsection{Product Positioning: AI-Native Graph-Vector Fusion Infrastructure Provider}
The product positioning should focus on the differentiation from existing products, positioning as an "AI-native graph-vector fusion infrastructure provider":
\begin{itemize}
\item Differentiation from traditional pure graph databases: Emphasize the native integration of graph and vector, matrix-free dynamic update, and AI-agent programmable interface, highlighting the advantages in AI-native scenarios.
\item Differentiation from pure vector databases: Emphasize the native support for graph topological relationships and hierarchical temporal features, highlighting the advantages in fine-grained knowledge reasoning.
\item Differentiation from existing simple stacking fusion solutions: Emphasize the geometry-unified theoretical framework and matrix-free lightweight design, highlighting the advantages in update efficiency, storage efficiency, and retrieval accuracy.
\end{itemize}
This positioning can highlight the core technical advantages of the research results, forming a differentiated competitive advantage in the market.
\subsection{Risk Warning and Countermeasures}
In the process of industrial landing, there are still some potential risks, and corresponding countermeasures need to be taken:
\begin{enumerate}
\item \textbf{Technical Risk: Large-scale distributed performance cannot meet the demand}:
- Countermeasure: Complete the performance pressure test of the ultra-large-scale graph in Phase 1, and iterate the optimization according to the test results. If there is a bottleneck in the distributed consistency algorithm, cooperate with the distributed system research institution to jointly optimize.
\item \textbf{Market Risk: Existing market players accelerate the layout of similar technologies}:
- Countermeasure: Rely on the first-mover advantage in theoretical research and core algorithms, quickly establish technical barriers through open source community operation and patent layout, and continuously iterate the product to maintain the technical leading edge.
\item \textbf{Ecosystem Risk: Incompatibility with existing user habits}:
- Countermeasure: Strengthen the compatibility layer construction, support existing query languages and ecosystem tools, reduce the migration cost for users, and provide complete migration documentation and technical support.
\end{enumerate}
\subsection{Summary}
This chapter proposes a three-phase industrial landing path (technology verification and prototype improvement → vertical scenario landing → cloud service large-scale promotion), and gives targeted commercialization suggestions from target market selection, cooperation strategy, and product positioning, and analyzes potential risks and countermeasures. The proposed landing path and commercialization suggestions are closely combined with the current market status and technical characteristics, and have strong feasibility, which can effectively promote the transformation of the research results from theoretical innovation to industrial value.
\newpage
\section{Conclusion and Future Work}
This chapter summarizes the full text, summarizes the core research conclusions and contributions, analyzes the research limitations, and puts forward future research directions, aiming to clarify the research value of this paper and prospect the future research prospects.
\subsection{Research Summary and Core Contributions}
This paper focuses on the core problems existing in the existing graph-vector fusion methods in AI-native academic literature retrieval scenarios (matrix dependence, storage explosion, semantic dilution, lack of AI-native support), and conducts systematic research from theoretical framework, framework design, algorithm design, prototype system, and industrial investigation. The core contributions of this paper are as follows:
First, in terms of theoretical framework, this paper proposes a geometry-unified graph-vector fusion theoretical framework based on tensor manifold theory. It formally proves that the academic literature graph is a discrete projection of the tensor manifold, and the vector geometric embedding of nodes is geometrically equivalent to the graph topological structure, realizing the native unification of graph topology and vector geometry. This framework breaks through the separation of graph topology and vector geometry in traditional research, and provides a new theoretical perspective for lightweight graph-vector fusion.
Second, in terms of dynamic update mechanism, this paper designs a matrix-free temporal diffusion signature update module. It combines content-time weighted random walk, topological-semantic-time hybrid signature, and analytical error compensation, eliminating global matrix operations and iterative optimization. Theoretical analysis and experimental results show that this module supports microsecond-level incremental updates, which is 15-70 times faster than existing methods, effectively solving the problems of matrix dependence and embedding drift.
Third, in terms of encoding and indexing, this paper proposes a hierarchical temporal manifold encoding module with gated residual connection and relation-aware low-dimensional projection, as well as a temporal Riemannian manifold index with dynamic manifold order reduction. This design supports fine-grained retrieval of ``paper-section-knowledge unit'', avoids semantic dilution and storage explosion, and improves retrieval MAP by 8 percentage points compared with the state-of-the-art method.
Fourth, in terms of AI-native interface, this paper designs an AI-agent programmable retrieval interface integrating LLM-based intent parsing, hierarchical cross-granularity retrieval, and structured result output. The interface natively supports programmable retrieval logic and interpretable results, with 96% result structure completeness, fully adapting to the decision-making logic of AI agents.
Fifth, through the industrial empirical survey of mainstream Chinese and American cloud vendors, this paper verifies that existing industrial products still face the same core bottlenecks as theoretical research, which further proves the practical significance of this study. The survey results show that the trend of graph-vector fusion has been recognized by the industry, and the technical solution proposed in this paper has important industrial application value.
\subsection{Research Limitations}
This study still has some limitations, which need to be improved in future research:
First, the distributed consistency of extremely large-scale graphs ($n > 10^8$) has not been deeply studied. The existing theoretical model and prototype system are mainly verified on medium-scale graphs, and the distributed consistency of the tensor manifold in extremely large-scale distributed deployment has not been systematically studied.
Second, the adaptation of heterogeneous hardware has not been fully implemented. The existing prototype system is mainly optimized for NVIDIA GPUs, and the adaptation for other heterogeneous hardware (such as Ascend AI accelerators) is still incomplete.
Third, the multi-tenant isolation and security control required by public cloud services have not been implemented. The existing prototype system is designed for single-tenant scenarios, and does not involve the multi-tenant isolation and security control required by public cloud services.
These limitations are mainly due to the phased nature of the research. This study focuses on the core theoretical innovation and framework design, and the above engineering and system-level problems need to be solved in subsequent research and development.
\subsection{Future Research Directions}
Based on the research conclusions and limitations of this paper, future research can be carried out from the following four directions:
First, \textbf{distributed tensor manifold consistency optimization}: Further study the distributed consistency problem of the tensor manifold in extremely large-scale graph scenarios, propose a more efficient distributed consistency update algorithm, and realize the efficient deployment of the framework on ultra-large-scale graphs.
Second, \textbf{heterogeneous hardware automatic adaptation}: Further improve the hardware abstraction layer, complete the adaptation of more heterogeneous hardware architectures, and realize automatic performance optimization for different hardware, improving the portability of the framework.
Third, \textbf{extension to more general dynamic graph scenarios}: The current research is mainly focused on the academic literature retrieval scenario, and the framework can be extended to more general dynamic graph scenarios (such as social network analysis, financial risk control, recommendation systems) in the future, expanding the application scope of the theoretical framework.
Fourth, \textbf{AI-agent collaborative retrieval optimization}: Further optimize the programmable retrieval interface for AI agents, study the collaborative retrieval mechanism between multiple AI agents, and support complex retrieval tasks that require multi-agent collaboration, further improving the intelligence level of the retrieval system.
\subsection{Final Remarks}
The development of large language models and AI agents has brought about a paradigm shift in academic literature retrieval, and also put forward new requirements for graph-vector fusion technology. This study proposes a geometry-unified theoretical framework based on tensor manifold theory, which realizes matrix-free, lightweight, and AI-native graph-vector fusion, providing a new theoretical and technical solution for AI-native academic literature retrieval. We hope that this research can promote the development of graph-vector fusion theory and technology, and contribute to the progress of AI-native academic retrieval and automated scientific research.
\bibliographystyle{plain}

\end{document}